\documentclass[12pt]{article}

\usepackage[english]{babel}
\usepackage[utf8]{inputenc}
\usepackage{amsmath}
\usepackage{graphicx}
\usepackage[colorinlistoftodos]{todonotes}
\usepackage{amssymb}
\usepackage{mathtools}
\usepackage{amsthm}
\usepackage{natbib}
\usepackage{tikz}
\usepackage{bbm}
\usepackage{adjustbox}
\usepackage{etoolbox}
\usepackage{enumerate}
\usepackage{pdflscape}
\usepackage{adjustbox}
\usepackage{scrextend}
\usepackage{titling}
\usepackage{setspace}
\usepackage{subcaption}
\usepackage{float}
\usepackage{bm}
\usepackage{mathrsfs}
\usepackage[margin=1.25in]{geometry}
\usepackage[normalem]{ulem}
\usepackage{hyperref}
\usepackage{mathtools}
\usepackage{arydshln}

\DeclarePairedDelimiter\floor{\lfloor}{\rfloor}
\hypersetup{pdfborder = 0 0 0,%
    pdftitle = Bootstrapping Risk Measures,%
    pdfauthor = Beutner Heinemann Smeekes,%
}

\numberwithin{equation}{section}

\newcommand{\EE}{\mathbb{E}}
\newcommand{\PP}{\mathbb{P}}
\newcommand{\Cov}{\mathbb{C}\mbox{ov}}
\newcommand{\Var}{\mathbb{V}\mbox{ar}}

\newcommand{\Z}{\mathbb{Z}}
\newcommand{\R}{\mathbb{R}}
\newcommand{\N}{\mathbb{N}}

\newcommand{\iid}{\mbox{\textit{i.i.d.}}}

\theoremstyle{definition}
\newtheorem{assumption}{Assumption}
\newtheorem{assumption*}{Assumption}

\newtheorem{algorithm}{Algorithm}

\theoremstyle{plain}
\newtheorem{theorem}{Theorem}
\newtheorem{lemma}{Lemma}
\newtheorem{corollary}{Corollary}

\theoremstyle{remark}

\title{V1 A Residual Bootstrap for Conditional Expected Shortfall}

\author{Alexander Heinemann and Sean Telg}

\date{\today}

\begin{document}

\begin{center}

\scshape
\LARGE{\textbf{A Residual Bootstrap for Conditional Expected Shortfall}}
\normalfont
\end{center}
\begin{center}
\par\vspace{0.5cm}
\text{Alexander Heinemann$^\dagger$}
\hspace{3cm}
\text{Sean Telg$^\dagger$}
\par\vspace{0.5cm}
\text{$^\dagger$Department of Quantitative Economics}
\par
\text{Maastricht University}
\par\vspace{0.5cm}
\text{\today}
\vspace{1cm}
\end{center}

\begin{abstract}
This paper studies a fixed-design residual bootstrap method for the two-step estimator of \cite{francq2015risk} associated with the conditional Expected Shortfall. For a general class of volatility models the bootstrap is shown to be asymptotically valid under the conditions imposed by \cite{beutner2018residual}. A simulation study is conducted revealing that the average coverage rates are satisfactory for most settings considered. There is no clear evidence to have a preference for any of the three proposed bootstrap intervals. This contrasts results in \cite{beutner2018residual} for the VaR, for which the reversed-tails interval has a superior performance. 
\end{abstract}

\doublespacing

\section{Introduction}
\label{sec:5.1}

The assessment of market risk is a key challenge that financial market participants face on a daily basis.
To evaluate the risk, financial institutions primarily employ the risk measure Value-at-Risk (VaR) to meet the capital requirements enforced by the Basel Committee on Banking Supervision. Despite its popularity, the VaR is not a coherent risk measure as it fails to fulfill the subadditivity property \citep{artzneretal1999}. A coherent alternative is the related risk measure Expected Shortfall (ES). For a given level $\alpha$, it is defined as the expected return in the worst $100\alpha \%$ cases and is therefore sometimes called Expected Tail Loss.\footnote{In the literature, ES is also sometimes referred to as conditional VaR since it is defined as the expected loss given a VaR exceedence. Since conditional refers to temporal dependence (i.e.\ conditional on past returns) in this paper, we refrain from using this term to prevent any confusion.} In contrast to the VaR, the ES provides valuable information on the severity of an incurred loss, which makes it the preferred risk measure in practice (c.f.\ \citeauthor{acerbitasche2002a}, \citeyear{acerbitasche2002a}; \citeyear{acerbitasche2002b}). Consequently, the Basel Committee published revised standards in January 2016 resembling a shift from VaR towards ES as the underlying risk measure \citep{osmundsen2018}.

In the literature there is an increasing interest in conditional risk measures, which take into account the temporal dependence of asset returns. Frequently, the volatility dynamics are specified by a (semi-) parametric model such that the conditional ES can be expressed as the product of the conditional volatility and the ES of the innovations’ distribution (c.f.\  \citeauthor{francq2015risk}, \citeyear{francq2015risk}, Example 2). The latter can be treated as additional parameter, which is generally unknown just like the parameters of the conditional volatility model. Inferring the parameters from data leads to an evaluation of the conditional ES that is prone to estimation risk. As argued in \cite{beutner2018residual} this estimation uncertainty can be substantial for risk measures related to extreme events.

The uncertainty around point estimates is typically determined by asymptotic theory, in which one replaces unknown quantities in the limiting distribution by consistent estimates. For example, \cite{caiwang2008} and \cite{martinsetal2018} study the behavior of proposed nonparametric estimators for conditional VaR and ES, based on asymptotics and simulation studies. An alternative approach is
based on bootstrap approximations. Regarding the estimators of the volatility model’s parameters, several bootstrap methods have been examined, among which the sub-sample bootstrap \citep{hall2003inference}, the block bootstrap \citep{corradi2008bootstrap}, the wild bootstrap \citep{shimizu2009bootstrapping} and the residual bootstrap, both with recursive \citep{pascual2006bootstrap,hidalgo2007goodness} and fixed design (\citeauthor{shimizu2009bootstrapping}, \citeyear{shimizu2009bootstrapping}; \citeauthor{cavaliere2018fixed}, \citeyear{cavaliere2018fixed}, \citeauthor{beutner2018residual}, \citeyear{beutner2018residual}). However, the estimation of the conditional ES has received only limited attention in the bootstrap literature. \cite{christoffersen2005estimation} construct intervals for conditional ES based on a recursive-design residual bootstrap method. 
\cite{gao2008estimation} compare coverage probabilities for conditional ES based on this bootstrap method and asymptotic normality results in their simulation study.

In this paper, we extend results of \cite{beutner2018residual} derived for conditional VaR to the conditional ES estimator. In particular, we follow the two-step procedure of \cite{francq2015risk} for the estimation of the underlying parameters. In a first step, we obtain estimates of the parameters of the stochastic volatility model by quasi-maximum-likelihood (QML) estimation. Based on the model's residuals, an estimate for the innovations’ ES is obtained in the second step. We propose a fixed-design residual bootstrap method to mimic the finite sample distribution of this two-step estimator for a general class of volatility models. Moreover, an algorithm is provided for the construction of bootstrap intervals for the conditional ES. 

The remainder of the paper is organized as follows. Section \ref{sec:5.2} introduces a general class of volatility models and derives the conditional ES. The two-step estimation procedure is described in Section \ref{sec:5.3} and corresponding asymptotic results are provided under the assumptions imposed by \cite{beutner2018residual}. In Section \ref{sec:5.4}, a fixed-design residual bootstrap method is proposed and proven to be consistent. In addition, bootstrap intervals are constructed for the conditional ES. Section \ref{sec:5.5} consists of a Monte Carlo study. Section \ref{sec:5.6} concludes. Auxiliary results and proofs of the main results are gathered in the Appendix.

\section{Model}
\label{sec:5.2}
We consider conditional volatility models of the form
\begin{align}
\label{eq:5.2.1}
\epsilon_t = \sigma_t\eta_t,
\end{align}
with $t\in \Z$, where $\epsilon_t$ denotes the log-return, $\{\sigma_t\}$ is a volatility process and $\{\eta_t\}$ is a sequence of independent and identically distributed (\iid) variables. The volatility is assumed to be a measurable function of past observations
\begin{align}
\label{eq:5.2.2}
\sigma_{t+1}=\sigma_{t+1}(\theta_0)=\sigma(\epsilon_t,\epsilon_{t-1},\dots;\theta_0),
\end{align}
with $\sigma:\R^\infty\times \Theta\to(0,\infty)$ and $\theta_0$ denotes the true parameter vector belonging to the parameter space $\Theta \subset \R^r$, $r \in \N$. Various commonly used volatility models satisfy \eqref{eq:5.2.1}--\eqref{eq:5.2.2}; for examples see  \citeauthor{francq2015risk} (\citeyear{francq2015risk}, Table 1). Consider an arbitrary real-valued random variable $X$ (e.g.\ stock return) with cdf $F_X$. If $\EE_X[X^-]<\infty$ with $X^{-} = \max\{-X,0\}$, then the ES at level $\alpha \in (0,1)$ is finite and given by $ES_\alpha(X) = -\EE_X\big[X|X<F_X^{-1}(\alpha)\big]$.
Let $\mathcal{F}_{t-1}$ denote the $\sigma$-algebra generated by $\{\epsilon_u, u<t\}$. It follows that the conditional ES of $\epsilon_t$ given $\mathcal{F}_{t-1}$ at level $\alpha \in (0,1)$ is 
\begin{align}
\label{eq:5.2.3}
ES_\alpha(\epsilon_t|\mathcal{F}_{t-1}) =& \sigma (\epsilon_{t-1},\epsilon_{t-2},\dots;\theta_0) ES_\alpha(\eta_t).
\end{align}
As $\eta_t$ are \iid, the ES at level $\alpha$ of $\eta_t$ is constant for a given $\alpha$ and can be treated as a parameter. Setting $\mu_\alpha = -\EE\big[\eta_t|\eta_t<\xi_\alpha\big]$ with $\xi_\alpha = F^{-1}(\alpha)$ and $F$ denoting the cdf of $\eta_t$, \eqref{eq:5.2.3} reduces to 
\begin{align}
\label{eq:5.2.4}
ES_\alpha(\epsilon_t|\mathcal{F}_{t-1}) = \mu_\alpha \sigma_t(\theta_0).
\end{align}
Typically, $\alpha$ is chosen small (e.g.\ 5\% and 1\%) such that $\xi_\alpha<0$ and hence $\mu_\alpha>0$. Except for special cases\footnote{We derive the analytical expressions for $\mu_\alpha$ for the cases in which $\eta_t$ are normally as well as Student-$t$ distributed in Appendix B.}, $\mu_\alpha$ is unknown and needs to estimated just like $\theta_0$.

\section{Estimation}
\label{sec:5.3}
For the estimation of the parameters  $\theta_0$ and $\mu_\alpha$ we employ the two-step procedure of \citeauthor{francq2015risk} (\citeyear{francq2015risk}, Remark 3). First, the vector of the conditional volatility parameters $\theta_0$ is estimated by quasi-maximum-likelihood (QML). Since
\begin{align}
\label{eq:5.3.1}
\sigma_{t+1}(\theta) = \sigma(\epsilon_t,\epsilon_{t-1},\dots,\epsilon_{1}, \epsilon_{0},\epsilon_{-1},\dots;\theta)
\end{align}
can generally not be determined completely given a sample $\epsilon_1, ,\dots, \epsilon_n$, we replace the unknown presample observations by arbitrary values, say $\tilde{\epsilon}_t$, $t\leq 0$, yielding
\begin{align}
\label{eq:5.3.2}
\tilde{\sigma}_{t+1}(\theta) =& \sigma(\epsilon_t,\epsilon_{t-1},\dots,\epsilon_{1}, \tilde{\epsilon}_{0},\tilde{\epsilon}_{-1},\dots;\theta).
\end{align}
Then the QML estimator of $\theta_0$ is defined as a measurable solution $\hat{\theta}_n$ of 
\begin{align}
\label{eq:5.3.3}
\hat{\theta}_n=\arg\max_{\theta \in \Theta} \tilde{L}_n(\theta)
\end{align}
with the criterion function specified by
\begin{align*}
\tilde{L}_n(\theta) = \frac{1}{n}\sum_{t=1}^n \tilde{\ell}_t(\theta) \qquad \text{and} \qquad \tilde{\ell}_t(\theta)=-\frac{1}{2}\bigg(\frac{\epsilon_t}{\tilde{\sigma}_t(\theta)}\bigg)^2-\log \tilde{\sigma}_t(\theta).
\end{align*}
In the second step, $\mu_\alpha$ can be estimated on the basis of the first-step residuals, i.e. $\hat{\eta}_t=\epsilon_t/\tilde{\sigma}_t(\hat{\theta}_n)$. A reasonable estimator of $\mu_\alpha$ (c.f.\ \citeauthor{gao2008estimation}, \citeyear{gao2008estimation}) is given by
\begin{align}
\label{eq:5.3.4}
\hat{\mu}_{n,\alpha} = -\frac{\sum_{t=1}^n \hat{\eta}_t \mathbbm{1}_{\{\hat{\eta}_t<\hat{\xi}_{n,\alpha}\}}}{\sum_{t=1}^n  \mathbbm{1}_{\{\hat{\eta}_t<\hat{\xi}_{n,\alpha}\}}},
\end{align}
where $\hat{\xi}_{n,\alpha}$ is the empirical $\alpha$-quantile of $\hat{\eta}_1,\dots,\hat{\eta}_n$, i.e.\ $\hat{\xi}_{n,\alpha}=\hat{\mathbbm{F}}_n^{-1}(\alpha)$ with empirical distribution function  $\hat{\mathbbm{F}}_n(x)=\frac{1}{n}\sum_{t=1}^n \mathbbm{1}_{\{\hat{\eta}_t\leq x\}}$.

Having obtained estimators for $\theta_0$ and $\mu_\alpha$, we turn to the estimation of the conditional ES of the one-period ahead observation at level $\alpha$. For notational convenience, we use the abbreviation $ES_{n,\alpha}$ to denote $ES_{\alpha}(\epsilon_{n+1}|\mathcal{F}_n)$. Employing \eqref{eq:5.3.2}--\eqref{eq:5.3.4} we can estimate $ES_{n,\alpha}$ by
\begin{align}
\label{eq:5.3.5}
\widehat{ES}_{n,\alpha}=\hat{\mu}_{n,\alpha}\: \tilde{\sigma}_{n+1}\big(\hat{\theta}_n\big).
\end{align}
For the asymptotic analysis of \eqref{eq:5.3.3}--\eqref{eq:5.3.5} we assume the conditions of \cite{beutner2018residual}, which we restate for completeness.
\begin{assumption}{(Compactness)}
\label{as:5.1}
$\Theta$ is a compact subset of $\R^r$.
\end{assumption}

\begin{assumption}{(Stationarity \& Ergodicity)}
\label{as:5.2}
$\{\epsilon_t\}$ is a strictly stationary and ergodic solution of \eqref{eq:5.2.1} with \eqref{eq:5.2.2}.
\end{assumption}

\begin{assumption}{(Volatility process)}
\label{as:5.3}
For any real sequence $\{x_i\}$, the function $\theta\to\sigma(x_1,x_2,\dots;\theta)$ is continuous.  Almost surely, $\sigma_t(\theta)>\underline{\omega}$ for any $\theta \in \Theta$ and some $\underline{\omega}>0$ and $\EE[\sigma_t^s(\theta_0)]<\infty$ for some $s>0$. Moreover, for any $\theta \in \Theta$, we assume $\sigma_t(\theta_0)/\sigma_t(\theta)=1$ almost surely (a.s.) if and only if $\theta=\theta_0$.
\end{assumption}

\begin{assumption}{(Initial conditions)}
\label{as:5.4}
There exists a constant $\rho \in (0,1)$ and a random variable $C_1$ measurable with respect to $\mathcal{F}_0$ and $\EE[|C_1|^s]<\infty$ for some $s>0$ such that
\begin{enumerate}[(i)]
\item \label{as:5.4.1} $\sup_{\theta \in \Theta}|\sigma_t(\theta)-\tilde{\sigma}_t(\theta)|\leq C_1 \rho^t$;

\item \label{as:5.4.2} $\theta\to \sigma(x_1, x_2, \dots;\theta)$ has continuous second-order derivatives satisfying
\begin{align*}
\sup_{\theta \in \Theta}\bigg|\bigg|\frac{\partial \sigma_t(\theta)}{\partial \theta}-\frac{\partial \tilde{\sigma}_t(\theta)}{\partial \theta}\bigg|\bigg|\leq  C_1 \rho^t, \qquad \quad \sup_{\theta \in \Theta}\bigg|\bigg|\frac{\partial^2 \sigma_t(\theta)}{\partial \theta \partial \theta'}-\frac{\partial^2 \tilde{\sigma}_t(\theta)}{\partial \theta\partial \theta'}\bigg|\bigg|\leq C_1 \rho^t,
\end{align*}
where $||\cdot||$ denotes the Euclidean norm.
\end{enumerate}
\end{assumption}

\begin{assumption}{(Innovation process)} 
\label{as:5.5}
The innovations $\{\eta_t\}$ satisfy

\begin{enumerate}[(i)]
\item  \label{as:5.5.1}  $\eta_t\overset{iid}{\sim}F$ with $F$ being continuous, $\EE\big[\eta_t^2\big]=1$ and $\eta_t$ is independent of $\{\epsilon_u:u<t\}$;

\item  \label{as:5.5.2} $\eta_t$ admits a density $f$ which is continuous and strictly positive around $\xi_\alpha<0$; 
\item  \label{as:5.5.3}  $\EE\big[\eta_t^{4}\big]<\infty$.

\end{enumerate}
\end{assumption}

\begin{assumption}{(Interior)}
 \label{as:5.6}
$\theta_0$ belongs to the interior of $\Theta$ denoted by $\mathring{\Theta}$.
\end{assumption}

\begin{assumption}{(Non-degeneracy)}
\label{as:5.7}
There does not exist a non-zero $\lambda\in \R^r$ such that $\lambda'\frac{\partial \sigma_t(\theta_0)}{\partial \theta}=0$ almost surely.
\end{assumption}

\begin{assumption}{(Monotonicity)}
\label{as:5.8}
For any real sequence $\{x_i\}$ and for any $\theta_1,\theta_2 \in \Theta$ satisfying $\theta_1\leq \theta_2$ componentwise, we have $\sigma(x_1,x_2,\dots;\theta_1)\leq \sigma(x_1,x_2,\dots;\theta_2)$.
\end{assumption}

\begin{assumption}{(Moments)}
\label{as:5.9}
There exists a neighborhood $\mathscr{V}(\theta_0)$ of $\theta_0$ such that the following variables have finite expectation
\begin{align*}
\text{(i)}\sup_{\theta \in \mathscr{V}(\theta_0)}\bigg|\frac{ \sigma_t(\theta_0)}{\sigma_t(\theta)}\bigg|^a, \qquad \;\;\: \text{(ii)} \sup_{\theta \in \mathscr{V}(\theta_0)}\bigg|\bigg|\frac{1}{\sigma_t(\theta)}\frac{\partial \sigma_t(\theta)}{\partial \theta}\bigg|\bigg|^{b}, \qquad \;\;\: \text{(iii)} \sup_{\theta \in \mathscr{V}(\theta_0)}\bigg|\bigg|\frac{1}{\sigma_t(\theta)}\frac{\partial^2 \sigma_t(\theta)}{\partial \theta \partial \theta'}\bigg|\bigg|^c
\end{align*}
for some $a$, $b$, $c$ (to be specified).
\end{assumption}

\begin{assumption}{(Scaling Stability)}
\label{as:5.10}
 There exists a function $g$ such that for any $\theta \in \Theta$, for any $\lambda>0$, and any real sequence $\{x_i\}$
\begin{align*}
\lambda \sigma(x_1,x_2,\dots;\theta)=\sigma(x_1,x_2,\dots;\theta_\lambda),
\end{align*}
where $\theta_\lambda=g(\theta,\lambda)$ and $g$ is differentiable in $\lambda$.
\end{assumption}
For a discussion of the conditions we refer to \citeauthor{francq2015risk} (\citeyear{francq2015risk}, Section 2 and 3) and \citeauthor{beutner2018residual} (\citeyear{beutner2018residual}, Section 3). On the basis of the previous assumptions we extend the strong consistency result of \citeauthor{francq2015risk} (\citeyear{francq2015risk}, Theorem 1) to the estimator of the ES at level $\alpha$ of $\eta_t$.
%
%
\begin{theorem}\textit{(Strong Consistency)}
\label{thm:5.1} Under Assumptions \ref{as:5.1}--\ref{as:5.3}, \ref{as:5.4}(\ref{as:5.4.1}) and \ref{as:5.5}(\ref{as:5.5.1}) the estimator  in \eqref{eq:5.3.3} is strongly consistent, i.e. $\hat{\theta}_n \overset{a.s.}{\to} \theta_0$.
If in addition Assumptions \ref{as:5.5}(\ref{as:5.5.3}), \ref{as:5.6},  and \ref{as:5.9}(i) hold with $a=-1,4$, then the  estimator in \eqref{eq:5.3.4} satisfies $\hat{\mu}_{n,\alpha} \overset{a.s.}{\to} \mu_\alpha$.
\end{theorem}
\begin{proof}
\citeauthor{francq2015risk} (\citeyear{francq2015risk}, Theorem 1) establish $\hat{\theta}_n \overset{a.s.}{\to} \theta_0$. Moreover, we have
\begin{align*}
\hat{\mu}_{n,\alpha} = -\frac{\frac{1}{n}\sum_{t=1}^n \hat{\eta}_t \mathbbm{1}_{\{\hat{\eta}_t<\hat{\xi}_{n,\alpha}\}}}{\frac{1}{n}\sum_{t=1}^n  \mathbbm{1}_{\{\hat{\eta}_t<\hat{\xi}_{n,\alpha}\}}} \overset{a.s.}{\to} -\frac{\EE[ \eta_t \mathbbm{1}_{\{\eta_t<\xi_\alpha\}}]}{\PP[\eta_t<\xi_{\alpha}]}= - \EE\big[\eta_t |\eta_t<\xi_\alpha\big]=\mu_\alpha
\end{align*}
by \citeauthor{beutner2018residual} (\citeyear{beutner2018residual}, Lemma 2), which verifies the second claim.
\end{proof}
To lighten notation, we henceforth write $D_t(\theta) =\frac{1}{\sigma_t(\theta)}\frac{\partial\sigma_t(\theta)}{\partial \theta}$  and drop the argument when evaluated at the true parameter, i.e.\ $D_t=D_t(\theta_0)$. The next result provides the joint asymptotic distribution of $\hat{\theta}_n$ and $\hat{\mu}_{n,\alpha}$ and is due to \cite{francq2012risk}.
\begin{theorem}\textit{(Asymptotic Distribution)}
\label{thm:5.2} Suppose Assumptions \ref{as:5.1}--\ref{as:5.7}, \ref{as:5.9} and \ref{as:5.10} hold with $a=b=4$ and $c=2$. Then, we have
\begin{align}
\label{eq:5.3.6}
      \begin{pmatrix}
      \sqrt{n}(\hat{\theta}_n-\theta_0)\\
    \sqrt{n}(\hat{\mu}_{n,\alpha} - \mu_\alpha)
      \end{pmatrix}
\overset{d}{\to}N\big(0, \Gamma_\alpha\big) \qquad \mbox{with}\qquad
\Gamma_\alpha=
      \begin{pmatrix}
      \frac{\kappa-1}{4}J^{-1} & \varphi_\alpha J^{-1}\Omega\\
    \varphi_\alpha \Omega'J^{-1} & \nu_\alpha
      \end{pmatrix},
\end{align}
where $\kappa = \EE[\eta_t^4]$, $\Omega = \EE[D_t]$, $J=\EE[D_tD_t']$, $\varphi_\alpha = \frac{1}{2}x_\alpha-\mu_\alpha \frac{\kappa-1}{4}$, $\nu_\alpha = \sigma_\alpha^2 - x_\alpha \mu_\alpha +\frac{\kappa-1}{4}\mu_\alpha^2$,
$\sigma_\alpha^2 = \frac{1}{\alpha^2}\Var[\big(\eta_t -\xi_\alpha\big)\mathbbm{1}_{\{\eta_t< \xi_\alpha\}}]$ and $x_\alpha = -\frac{1}{\alpha}\Cov\big[\eta_t^2, (\eta_t-\xi_\alpha) \mathbbm{1}_{\{\eta_t<\xi_\alpha\}}\big]$.
\end{theorem}
In order to evaluate $\sigma^{2}_{\alpha}$ and $x_{\alpha}$ in Theorem \ref{thm:5.2}, we need expressions for the variance and covariance term respectively. After basic manipulation we find
\begin{align*}
&\Var[\big(\eta_t -\xi_\alpha\big)\mathbbm{1}_{\{\eta_t< \xi_\alpha\}}] = p_{\alpha} + \alpha + \xi_{\alpha}(1-\alpha)\alpha(\xi_{\alpha}-2\mu_{\alpha}) - (\alpha\mu_{\alpha})^{2} \\   
&\Cov\big[\eta_t^2, (\eta_t-\xi_\alpha) \mathbbm{1}_{\{\eta_t<\xi_\alpha\}}\big] = -\alpha\mu_{\alpha} - (\xi_{\alpha}p_{\alpha} - q_{\alpha}),    
\end{align*}
with $p_{\alpha} = \EE[\eta_{t}^{2}\mathbbm{1}_{\{\eta_t<\xi_\alpha\}}] - \alpha$ and $q_\alpha = \EE[\eta_t^{3}\mathbbm{1}_{\{\eta_t< \xi_\alpha\}}]$. In a GARCH($p,q$) setting, \cite{gao2008estimation} quantify the uncertainty around $\hat{\theta}_n$ and $\hat{\mu}_{n,\alpha}$ using \eqref{eq:5.3.6} while replacing  the unknown quantities in $\Gamma_\alpha$ by  estimates. In the same spirit, $\xi_\alpha$ and $\mu_\alpha$ can be substituted by $\hat{\xi}_{n,\alpha}$ and $\hat{\mu}_{n,\alpha}$ while $\Omega$, $J$, $q_\alpha$, $p_\alpha$ and $\kappa$ can be replaced by
\begin{align}
\label{eq:5.3.7}
\begin{split}
\hat{\Omega}_n=&\frac{1}{n}\sum_{t=1}^n\hat{D}_t, \qquad   \qquad \: \hat{q}_{n,\alpha}=\frac{1}{n}\sum_{t=1}^n \hat{\eta}_t^3 \mathbbm{1}_{\{\hat{\eta}_t<\hat{\xi}_{n,\alpha}\}}, \qquad \quad   \hat{\kappa}_n=\frac{1}{n}\sum_{t=1}^n\hat{\eta}_t^4,\\
\hat{J}_n=&\frac{1}{n}\sum_{t=1}^n\hat{D}_t\hat{D}_t', \qquad  \quad  \hat{p}_{n,\alpha}=\frac{1}{n}\sum_{t=1}^n \hat{\eta}_t^2 \mathbbm{1}_{\{\hat{\eta}_t<\hat{\xi}_{n,\alpha}\}}-\alpha,
\end{split}
\end{align}
with $\hat{D}_t = \tilde{D}_t(\hat{\theta}_n)$ and $\tilde{D}_t(\theta) = \frac{1}{\tilde{\sigma}_t(\theta)}\frac{\partial\tilde{\sigma}_t(\theta)}{\partial \theta}$. The strong consistency of the estimators in \eqref{eq:5.3.7} follows from \citeauthor{beutner2018residual} (\citeyear{beutner2018residual}, Lemma 2 and Theorem 1). Based on \eqref{eq:5.3.7} we obtain a consistent estimator for $\Gamma_\alpha$ denoted by $\hat{\Gamma}_{n,\alpha}$. Note that in the joint asymptotic distribution in Theorem \ref{thm:5.2} the pdf of $\eta_{t}$ does not occur. This is in contrast to the limiting distribution of the parameters that comprise the conditional VaR estimator \citep{beutner2018residual}. Hence, no density estimation (by e.g.\ kernel smoothing) is required here.

The asymptotic behavior of the conditional ES estimator can be studied by employing Theorem \ref{thm:5.2}. Since the conditional volatility varies over time, a limiting distribution cannot exist and therefore the concept of weak convergence is not applicable in this context. \citeauthor{beutner2017justification} (\citeyear{beutner2017justification}, Section 4) advocate a \textit{merging} concept that generalizes the notion of weak convergence, i.e.\ two sequences of (random) probability measures $\{P_n\},\{Q_n\}$ \textit{merge} (in probability) if and only if their bounded Lipschitz distance $d_{BL}(P_n,Q_n)$ converges to zero (in probability). Assuming two independent samples, one for parameter estimation and one for conditioning, the delta method suggests that the ES estimator, centered at $ES_{n,\alpha}$ and inflated by $\sqrt{n}$, and
\begin{align}
\label{eq:5.3.8}
N\left(0, \begin{pmatrix}
    \mu_\alpha \frac{\partial \sigma_{n+1}(\theta_0)}{\partial \theta}\\
    \sigma_{n+1}
      \end{pmatrix}' \Gamma_\alpha \begin{pmatrix}
    \mu_\alpha \frac{\partial \sigma_{n+1}(\theta_0)}{\partial \theta}\\
         \sigma_{n+1}
      \end{pmatrix}\right)
\end{align}
given $\mathcal{F}_n$ merge in probability. Equation \eqref{eq:5.3.8} highlights once more the relevance of the merging concept since the conditional variance still depends on $n$ and does not converge as $n \to \infty$. In combination with Theorem \ref{thm:5.1} and $\hat{\Gamma}_{n,\alpha} \overset{a.s.}{\to} \Gamma_\alpha$, $100(1-\gamma)\%$ confidence intervals for $ES_{n,\alpha}$ can be constructed with bounds given by
\begin{align}
\label{eq:5.3.9}
\widehat{ES}_{n,\alpha}\pm \frac{\Phi^{-1}(\gamma/2)}{\sqrt{n}}
\left\{\begin{pmatrix}
    \hat{\mu}_{n,\alpha} \frac{\partial \tilde{\sigma}_{n+1}(\hat{\theta}_n)}{\partial \theta}\\
         \tilde{\sigma}_{n+1}(\hat{\theta}_n)
      \end{pmatrix}' \hat{\Gamma}_{n,\alpha} \begin{pmatrix}
    \hat{\mu}_{n,\alpha} \frac{\partial \tilde{\sigma}_{n+1}(\hat{\theta}_n)}{\partial \theta}\\
    \tilde{\sigma}_{n+1}(\hat{\theta}_n)
      \end{pmatrix}\right\}^{1/2},
\end{align}
where $\Phi$ denotes the standard normal cdf. It has to be mentioned that researchers rarely have a replicate, independent of the original series, to their disposal.\footnote{Exceptions would include some experimental settings.} An asymptotic justification for the interval on the basis of a single sample is given in \cite{beutner2017justification}. Bootstrap methods offer an alternative way to quantify the uncertainty around the estimators.

\section{Bootstrap}
\label{sec:5.4}

\subsection{Fixed-Design Residual Bootstrap}
\label{sec:5.4.1}

We propose a fixed-design residual bootstrap procedure, described in Algorithm \ref{alg:5.1}, to approximate the distribution of the estimators in \eqref{eq:5.3.3}-\eqref{eq:5.3.5}.

\begin{algorithm}\textit{(Fixed-design residual bootstrap)}
\label{alg:5.1}
\begin{enumerate}

\item For $t=1,\dots, n$, generate  $\eta_t^* \overset{iid}{\sim} \hat{\mathbbm{F}}_n$ and the bootstrap observation $\epsilon_t^* = \tilde{\sigma}_t(\hat{\theta}_n) \eta_t^*$, where $\tilde{\sigma}_t(\theta)$ and $\hat{\theta}_n$ are given in \eqref{eq:5.3.2} and
\eqref{eq:5.3.3}, respectively.
\item Calculate the bootstrap estimator 
\begin{align}
\label{eq:5.4.1}
\hat{\theta}_n^* = \arg \max_{\theta \in \Theta}L_n^*(\theta)
\end{align}
with the bootstrap criterion function given by
\begin{align*}
L_n^*(\theta) = \frac{1}{n}\sum_{t=1}^n \ell_t^*(\theta) \qquad \text{and} \qquad \ell_t^*(\theta)=-\frac{1}{2}\bigg(\frac{\epsilon_t^{*}}{\tilde{\sigma}_t(\theta)}\bigg)^2-\log \tilde{\sigma}_t(\theta).
\end{align*}
\item For $t=1,\dots,n$ compute the bootstrap residual $\hat{\eta}_t^* = \epsilon_t^*/\tilde{\sigma}_t(\hat{\theta}_n^*)$
and obtain 
\begin{align}
\label{eq:5.4.2}
\hat{\mu}_{n,\alpha}^* =& -\frac{\sum_{t=1}^n\hat{\eta}_t^* \mathbbm{1}_{\{\hat{\eta}_t^*<\hat{\xi}_{n,\alpha}^*\}}}{\sum_{t=1}^n \mathbbm{1}_{\{\hat{\eta}_t^*<\hat{\xi}_{n,\alpha}^*\}}},
\end{align}
 where $\hat{\xi}_{n,\alpha}^*$ is the empirical $\alpha$-quantile of $\hat{\eta}_1^*,\dots,\hat{\eta}_n^*$.

\item Obtain the bootstrap estimator of the conditional ES
\begin{align}
\label{eq:5.4.3}
\widehat{ES}_{n,\alpha}^{*}=\hat{\mu}_{n,\alpha}^{*}\: \tilde{\sigma}_{n+1}\big(\hat{\theta}_n^{*}\big).
\end{align}
\end{enumerate}
\end{algorithm}
\noindent In the following subsection we show the asymptotic validity of the fixed-design bootstrap procedure described in Algorithm \ref{alg:5.1}. 

\subsection{Bootstrap Consistency}
\label{sec:5.4.2}

Subsequently, we employ the usual notation for bootstrap asymptotics, i.e.\ “$\overset{p^*}{\to}$" and “$\overset{d^*}{\to}$", as well as the standard bootstrap stochastic order symbol “$o_{p^*}(1)$" (c.f.\ \citeauthor{chang2003sieve}, \citeyear{chang2003sieve}). The asymptotic validity of the bootstrap corresponding to the stochastic volatility part is shown in \citeauthor{beutner2018residual} (\citeyear{beutner2018residual}, Proposition 1). Therefore, we focus only on $\hat{\mu}_{n,\alpha}^*$. By construction, we have $\sum_{t=1}^n \mathbbm{1}_{\{\hat{\eta}_t^*<\hat{\xi}_{n,\alpha}^*\}}=\floor{\alpha n}+1$, where $\floor{x}$ denotes the largest integer not exceeding $x$. Defining $\alpha_n=\frac{\floor{\alpha n}+1}{n}$, we standardize \eqref{eq:5.4.2} such that the bootstrap estimator satisfies
\begin{align}
\label{eq:5.4.4}
\begin{split}
    \sqrt{n}(\hat{\mu}_{n,\alpha}^*-\hat{\mu}_{n,\alpha}) =& -\frac{1}{\alpha_n}\frac{1}{\sqrt{n}}\sum_{t=1}^n \Big(\hat{\eta}_t^* \mathbbm{1}_{\{\hat{\eta}_t^*<\hat{\xi}_{n,\alpha}^*\}}+\alpha_n \hat{\mu}_{n,\alpha}\Big)\\
    =& -\frac{1}{\alpha_n} \big(A_n^*+B_n^*+C_n^*+D_n^*\big),
    \end{split}
\end{align}
where the scaling factor $-1/\alpha_n$ in \eqref{eq:5.4.4} converges to $-1/\alpha$ since $\alpha\leq \alpha_n \leq \alpha +\frac{1}{n}$. The different terms in brackets are given by
{\allowdisplaybreaks
\begin{align*}
A_n^*&=\frac{1}{\sqrt{n}}\sum_{t=1}^n\big(\hat{\eta}_t^* -\hat{\xi}_{n,\alpha}\big)\big(\mathbbm{1}_{\{\hat{\eta}_t^*< \hat{\xi}_{n,\alpha}^*\}}-\mathbbm{1}_{\{\eta_t^*< \hat{\xi}_{n,\alpha}\}}\big), \\
B_n^* &=\hat{\xi}_{n,\alpha}\frac{1}{\sqrt{n}}\sum_{t=1}^n \big(\mathbbm{1}_{\{\hat{\eta}_t^*< \hat{\xi}_{n,\alpha}^*\}}-\alpha_n\big), \\
C_n^* &=\frac{1}{\sqrt{n}}\sum_{t=1}^n\big(\hat{\eta}_t^* -\eta_t^*\big)\mathbbm{1}_{\{\eta_t^*< \hat{\xi}_{n,\alpha}\}}, \\
D_n^* &=\frac{1}{\sqrt{n}}\sum_{t=1}^n\Big(\big(\eta_t^* -\hat{\xi}_{n,\alpha}\big)\mathbbm{1}_{\{\eta_t^*< \hat{\xi}_{n,\alpha}\}}+\alpha_n \big( \hat{\xi}_{n,\alpha}+\hat{\mu}_{n,\alpha}\big)\Big).
\end{align*}}
Employing arguments of \citeauthor{chen2007nonparametric} (\citeyear{chen2007nonparametric}, Lemma 2) Lemma \ref{lem:5.1} in Appendix \ref{app:5.A} states the asymptotic negligibility of $A_n^*$, i.e.\ $A_n^*\overset{p^*}{\to}0$ in probability. The term $B_n^* = 0$ since $\frac{1}{n}\sum_{t=1}^n \mathbbm{1}_{\{\hat{\eta}_t^*<\hat{\xi}_{n,\alpha}^*\}}=\alpha_n$ by construction. Further, Lemma \ref{lem:5.2} in Appendix \ref{app:5.A} states that $C_n^*=\alpha \mu_\alpha \Omega \sqrt{n}\big(\hat{\theta}_n^*-\hat{\theta}_n\big)+o_{p^*}(1)$ in probability. Last, we have $D_n^*\overset{d^*}{\to}N(0,\nu_\alpha)$ almost surely by Lemma \ref{lem:5.3} in Appendix \ref{app:5.A}. The previous discussion together with the asymptotic expansion of $\sqrt{n}\big(\hat{\theta}_n^*-\hat{\theta}_n\big)$ in \citeauthor{beutner2018residual} (\citeyear{beutner2018residual}, Equation 4.4) yields
\begin{equation*}
      \resizebox{\textwidth}{!}{$\begin{pmatrix}
      \sqrt{n}(\hat{\theta}_n^*-\hat{\theta}_n)\\
    \sqrt{n}(\hat{\mu}_{n,\alpha}^* - \hat{\mu}_{n,\alpha})
      \end{pmatrix}  
      \text{=}\! 
      \begin{pmatrix}
      \frac{1}{2}J^{-1}&O_{r\times 1}\\
    -\frac{1}{2}\mu_\alpha\Omega'J^{-1} & -\frac{1}{\alpha}
      \end{pmatrix}\!\!\!
        \begin{pmatrix}
      \frac{1}{\sqrt{n}}\sum\limits_{t=1}^n \hat{D}_t\big(\eta_t^{*2}-1\big)\\
   \frac{1}{\sqrt{n}}\sum\limits_{t=1}^n\Big(\big(\eta_t^* -\hat{\xi}_{n,\alpha}\big)\mathbbm{1}_{\{\eta_t^*< \hat{\xi}_{n,\alpha}\}}+\alpha_n \big( \hat{\xi}_{n,\alpha}+\hat{\mu}_{n,\alpha}\big)\Big)
      \end{pmatrix}\!\! +\! o_{p^*}(1)$}
\end{equation*}
%
in probability. Employing Lemma \ref{lem:5.3} once more leads to the paper's main result.
\begin{theorem}\textit{(Boostrap consistency)}
\label{thm:5.3}
Suppose Assumptions \ref{as:5.1}--\ref{as:5.10} hold with $a=\pm 12$, $b=12$ and $c=6$. Then, we have
\begin{align*}
      \begin{pmatrix}
      \sqrt{n}(\hat{\theta}_n^*-\hat{\theta}_n)\\
    \sqrt{n}(\hat{\mu}_{n,\alpha}^* - \hat{\mu}_{n,\alpha})
      \end{pmatrix}
\overset{d^*}{\to}N\big(0, \Gamma_\alpha\big)
\end{align*}
in probability.
\end{theorem}
Theorem \ref{thm:5.3} is useful to validate the bootstrap for the conditional ES estimator. For the asymptotic behavior of the conditional ES estimator we refer to \eqref{eq:5.3.8} and the text preceding it. The following corollary is established.
\begin{corollary}
\label{cor:5.1}
Under the assumptions of Theorem \ref{thm:5.3} the conditional distribution of $\sqrt{n}\big(\widehat{ES}_{n,\alpha}^{*}-\widehat{ES}_{n,\alpha}\big)$ given $\mathcal{F}_n$ and \eqref{eq:5.3.8} given $\mathcal{F}_n$ merge in probability.
\end{corollary}
Having proven first-order asymptotic validity of the bootstrap procedure described in Section \ref{sec:5.4.1}, we turn to constructing bootstrap confidence intervals for ES.

\subsection{Bootstrap Confidence Intervals for ES}\label{sec:5.4.3}

Clearly, the ES evaluation in \eqref{eq:5.3.5} is subject to estimation risk that needs to be quantified. We propose the following algorithm to obtain approximately $100(1-\gamma)\%$ confidence intervals.  

\begin{algorithm}{\textit{(Fixed-design Bootstrap Confidence Intervals for ES)}}
\label{alg:5.2}
\begin{enumerate}

\item Acquire a set of $B$ bootstrap replicates, i.e.  $\widehat{ES}_{n,\alpha}^{*(b)}$ for $b=1,\dots,B$, by repeating Algorithm 1.

\item[2.1.] Obtain the \textit{equal-tailed percentile} (EP) interval
\begin{align}
\label{eq:77883241}
\bigg[\widehat{ES}_{n,\alpha}-\frac{1}{\sqrt{n}}\hat{G}_{n,B}^{*-1}(1-\gamma/2),\:\widehat{ES}_{n,\alpha}-\frac{1}{\sqrt{n}}\hat{G}_{n,B}^{*-1}(\gamma/2)\bigg]
\end{align}
with $\hat{G}_{n,B}^*(x)=\frac{1}{B}\sum_{b=1}^B \mathbbm{1}_{\big\{\sqrt{n}\big(\widehat{ES}_{n,\alpha}^{*(b)}-\widehat{ES}_{n,\alpha}\big)\leq x\big\}}$. 
\item[2.2.] Calculate the \textit{reversed-tails} (RT) interval
\begin{align}
\label{eq:77883242}
\bigg[\widehat{ES}_{n,\alpha}+\frac{1}{\sqrt{n}}\hat{G}_{n,B}^{*-1}(\gamma/2),\widehat{ES}_{n,\alpha}+\frac{1}{\sqrt{n}}\hat{G}_{n,B}^{*-1}(1-\gamma/2)\bigg].
\end{align}
\item[2.3.] Compute the \textit{symmetric} (SY) interval
\begin{align}
\label{eq:77883243}
\bigg[\widehat{ES}_{n,\alpha}-\frac{1}{\sqrt{n}}\hat{H}_{n,B}^{*-1}(1-\gamma),\:\widehat{ES}_{n,\alpha}+\frac{1}{\sqrt{n}}\hat{H}_{n,B}^{*-1}(1-\gamma)\bigg]
\end{align}
with $\hat{H}_{n,B}^*(x)=\frac{1}{B}\sum_{b=1}^B \mathbbm{1}_{\big\{\sqrt{n}\big|\widehat{ES}_{n,\alpha}^{*(b)}-\widehat{ES}_{n,\alpha}\big|\leq x\big\}}$.
\end{enumerate}
\end{algorithm}
For a discussion of the three interval types in Algorithm \ref{alg:5.2}, we refer to \citeauthor{beutner2018residual} (\citeyear{beutner2018residual}, Section 4.3). In the next section, features of the fixed-design bootstrap confidence intervals for the conditional ES are studied by means of simulations.

\section{Monte Carlo Experiment}
\label{sec:5.5}

To assess the proposed bootstrap procedure in finite samples, we consider a simulation setup similar to \cite{beutner2018residual}. The Data Generating Process (DGP) is a GARCH($1,1$), which falls in the class of conditional volatility models defined in \eqref{eq:5.2.1}--\eqref{eq:5.2.2}. More specifically, we consider
\begin{align*}
\begin{cases}
    &\epsilon_{t} = \sigma_{t}\eta_{t}, \\
    &\sigma_{t+1}^2 = \omega_0+ \alpha_0 \epsilon_t^2+ \beta_0 \sigma_t^2
\end{cases}
\end{align*}
with $\theta_0=(\omega_0,\alpha_0,\beta_0)^\prime$. Regarding the GARCH parameters we study two scenarios:
\begin{enumerate}[(i)]
    \item high persistence: $\theta_0=(0.05\times 20^2/252,0.15,0.8)^\prime$,
    \item low persistence: $\theta_0=(0.05\times 20^2/252,0.4,0.55)^\prime$.
\end{enumerate}
The innovations $\{\eta_t\}$ are drawn from two different distributions: the Student-$t$ distribution with $\nu=6$ degrees of freedom and the standard normal distribution (which corresponds to the case $\nu=\infty$). Whereas in the latter case the innovations are appropriately standardized, in the former we draw from the normalized density $f(x)=\frac{1}{\sigma_\nu}f_\nu(x/\sigma_\nu)$ such that $\EE[\eta_t^2]=1$, where $\sigma_\nu^2=\frac{\nu-2}{\nu}$ and $f_\nu(x)=\frac{\Gamma(\frac{\nu+1}{2})}{\sqrt{\nu \pi}\Gamma(\frac{\nu}{2})}\big(1+\frac{x^2}{\nu}\big)^{-\frac{\nu+1}{2}}$. In this setting, the ES of the innovations’ distribution reduces to $\mu_\alpha=\frac{f_{\nu-2}(\xi_\alpha)}{\alpha}$ with $\xi_\alpha=\sigma_\nu F_\nu^{-1}(\alpha)$ and $F_\nu(x)=\int_{-\infty}^x f_\nu(y) dy$; we refer to Appendix \ref{app:5.B} for details. 
For the experiment, the ES level takes two values: $\alpha \in \{0.01, 0.05 \}$. We consider four different sample sizes $n \in \{ 500; 1{,}000; 5{,}000; 10{,}000 \}$ and the number of bootstrap replicates is fixed at $B = 2{,}000$. For each model, we simulate $S = 2{,}000$ independent Monte Carlo trajectories. All simulations are carried out on a HP Z640 workstation with 16 cores using Matlab R2016a. The numerical optimization of the log-likelihood function is performed using the built-in function \textit{fmincon}. Parallel computing by means of \textit{parfor} is employed to reduce running time significantly.

\citeauthor{beutner2018residual} (\citeyear{beutner2018residual})
 demonstrate that the bootstrap distribution mimics adequately the finite sample distribution of the estimator of the volatility parameters.
In a similar fashion, we assess whether the bootstrap distribution (given a particular sample) mimics the distribution of the ES parameter estimator. 
\begin{figure}[tbp]
\centering
\begin{subfigure}[b]{0.41\textwidth}
	\centering
	\includegraphics[width=\textwidth]{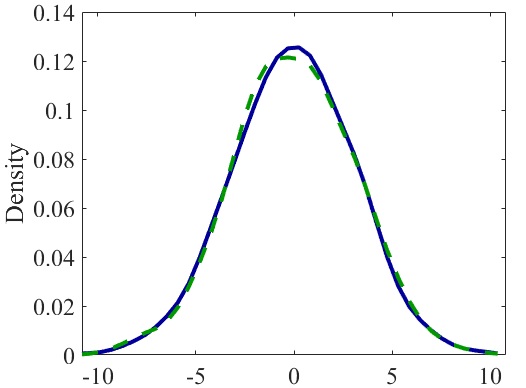}
            \caption{$\alpha = 0.05$}    
\end{subfigure}
\quad
\begin{subfigure}[b]{0.41\textwidth}  
	\centering 
	\includegraphics[width=\textwidth]{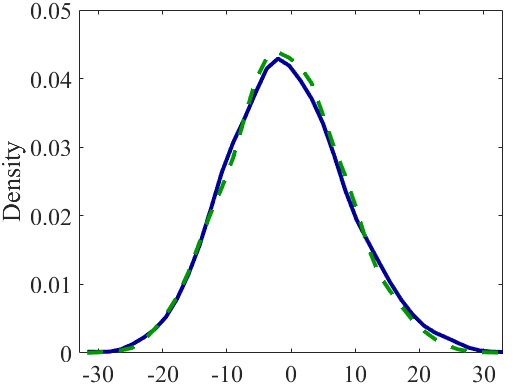}
	\caption{$\alpha = 0.01$} 
\end{subfigure}
\caption{Density estimates for the distribution of $\sqrt{n}(\hat{\mu}_{n\alpha}-\mu_\alpha)$ (full line) based on $S=2{,}000$ simulations and the fixed-design bootstrap distribution of $\sqrt{n}(\hat{\mu}_{n,\alpha}^*-\hat{\mu}_{n,\alpha})$ (dashed line) based on $B=2{,}000$ replications. The DGP is a GARCH($1,1$) with $\theta_0 = (0.08,0.15,0.8)'$, sample size $n=5,000$ and (normalized) Student-$t$ innovations (6 degrees of freedom).} 
\label{fig:5.1}
\end{figure}
Figure \ref{fig:5.1} displays the density estimates for the distribution of $\sqrt{n}(\hat{\mu}_{n,\alpha}-\mu_\alpha)$ and $\sqrt{n}(\hat{\mu}_{n,\alpha}^*-\hat{\mu}_{n,\alpha})$ in the high persistence case for $n=5{,}000$ with $\alpha \in \{0.01,0.05\}$. In both cases, we observe that the density plots are bell curves around the value zero, which supports the theoretical results of Theorem \ref{thm:5.2} and \ref{thm:5.3}. Since the density graphs for the other scenarios are very similar, they are not reported in order to conserve space. We continue by studying the coverage probabilities of the three bootstrap intervals introduced in Section \ref{sec:5.4.3}. 

\begin{table}[h]
\resizebox{\textwidth}{!}{\begin{tabular}{rccccccc}
\hline \hline
\multicolumn{1}{c}{\begin{tabular}[c]{@{}c@{}}Sample\\ size\end{tabular}} &                         & \begin{tabular}[c]{@{}c@{}}Average\\ coverage\end{tabular} & \begin{tabular}[c]{@{}c@{}}Av. coverage\\ below/above\end{tabular} & \begin{tabular}[c]{@{}c@{}}Average\\ length\end{tabular} & \begin{tabular}[c]{@{}c@{}}Average\\ coverage\end{tabular} & \begin{tabular}[c]{@{}c@{}}Av. coverage\\ below/above\end{tabular} & \begin{tabular}[c]{@{}c@{}}Average\\ length\end{tabular} \\ \hline
\multicolumn{1}{c}{}                                                      &                         &                                                            &                                                                    &                                                          &                                                            &                                                                    &                                                          \\
\multicolumn{1}{c}{}                                                      &                         & \multicolumn{3}{c}{low persistence}                                                                                                                                                        & \multicolumn{3}{c}{high persistence}                                                                                                                                                       \\
$500$                                                                     & \multicolumn{1}{c:}{EP} & $84.35$                                                    & $3.85$/$11.80$                                                     & \multicolumn{1}{c:}{$0.613$}                             & $83.15$                                                    & $4.45$/$12.40$                                                     & $0.825$                                                  \\
                                                                          & \multicolumn{1}{c:}{RT} & $85.30$                                                    & $2.35$/$12.35$                                                     & \multicolumn{1}{c:}{$0.613$}                             & $86.00$                                                    & $2.15$/$11.85$                                                     & $0.825$                                                  \\
                                                                          & \multicolumn{1}{c:}{SY} & $86.00$                                                    & $2.50$/$11.50$                                                     & \multicolumn{1}{c:}{$0.623$}                             & $85.45$                                                    & $2.65$/$11.90$                                                     & $0.833$                                                  \\ \hdashline
$1,000$                                                                   & \multicolumn{1}{c:}{EP} & $88.05$                                                    & $3.40$/$8.55$                                                      & \multicolumn{1}{c:}{$0.448$}                             & $87.55$                                                    & $3.95$/$8.50$                                                      & $0.610$                                                  \\
                                                                          & \multicolumn{1}{c:}{RT} & $88.70$                                                    & $2.70$/$8.60$                                                      & \multicolumn{1}{c:}{$0.448$}                             & $88.85$                                                    & $2.75$/$8.40$                                                      & $0.610$                                                  \\
                                                                          & \multicolumn{1}{c:}{SY} & $88.85$                                                    & $2.60$/$8.55$                                                      & \multicolumn{1}{c:}{$0.452$}                             & $88.55$                                                    & $3.00$/$8.45$                                                      & $0.612$                                                  \\ \hdashline
$5,000$                                                                   & \multicolumn{1}{c:}{EP} & $89.35$                                                    & $4.35$/$6.30$                                                      & \multicolumn{1}{c:}{$0.216$}                             & $89.85$                                                    & $3.35$/$6.80$                                                      & $0.287$                                                  \\
                                                                          & \multicolumn{1}{c:}{RT} & $89.90$                                                    & $3.75$/$6.35$                                                      & \multicolumn{1}{c:}{$0.216$}                             & $89.90$                                                    & $3.45$/$6.65$                                                      & $0.287$                                                  \\
                                                                          & \multicolumn{1}{c:}{SY} & $90.05$                                                    & $4.05$/$5.90$                                                      & \multicolumn{1}{c:}{$0.216$}                             & $90.15$                                                    & $3.20$/$6.65$                                                      & $0.287$                                                  \\ \hdashline
$10,000$                                                                  & \multicolumn{1}{c:}{EP} & $89.00$                                                    & $4.70$/$6.30$                                                      & \multicolumn{1}{c:}{$0.155$}                             & $89.60$                                                    & $4.35$/$6.05$                                                      & $0.204$                                                  \\
\multicolumn{1}{c}{}                                                      & \multicolumn{1}{c:}{RT} & $89.60$                                                    & $4.55$/$5.85$                                                      & \multicolumn{1}{c:}{$0.155$}                             & $89.55$                                                    & $4.55$/$5.90$                                                      & $0.204$                                                  \\
\multicolumn{1}{c}{}                                                      & \multicolumn{1}{c:}{SY} & $89.55$                                                    & $4.35$/$6.10$                                                      & \multicolumn{1}{c:}{$0.155$}                             & $89.55$                                                    & $4.45$/$6.00$                                                      & $0.204$                                                  \\ \hline \hline
\end{tabular}}
\caption{The table reports distinct features of the fixed-design bootstrap confidence intervals for the conditional ES at \textbf{level} $\bm{\alpha=0.05}$ with \textbf{nominal coverage} $\bm{1-\gamma=90\%}$. For each interval type and different sample sizes ($n$), the interval's average coverage rates (in $\%$), the average rate of the conditional ES being below/above the interval (in $\%$) and the interval's average length are tabulated. The intervals are based on $B=2{,}000$ bootstrap replications and the averages are computed using $S=2{,}000$ simulations. The DGP is a GARCH(1,1) with (normalized) \textbf{Student-$t$ innovations ($\bm{6}$ degrees of freedom)}.}
\label{tab:5.1}
\end{table}

Table \ref{tab:5.1} reports the results of the three $90\%$--bootstrap intervals for the $5\%$--ES with Student-$t$ distributed innovations (which we consider as benchmark). For moderate sample sizes, we observe satisfactory coverage probabilities that lie relatively close to the nominal level of $90\%$. For small sample size ($n=500)$, the intervals exhibit small under-coverage with values ranging from $4.00$ to $5.85$ percentage points ($pp$) below the nominal value.
%
For all three intervals, we find that the average rate of the conditional ES being below the interval is considerably less than it being above the interval. This phenomenon is most pronounced in smaller sample size.
Concerning the average length of the intervals, we can make two important observations. Firstly, the SY interval is generally larger than the EP/RT interval.\footnote{By construction, the EP and the RT interval are of equal length.} As sample size increases, this gap disappears and all intervals' average lengths shrink. Secondly, the average length of intervals is larger in the high persistence case, as the conditional volatility varies more compared to the lower persistence case. In the following, we study deviations from the benchmark specification. Table \ref{tab:5.2} considers a change in the innovation distribution $F$, while Table \ref{tab:5.3} and \ref{tab:5.4} take into account a change in the ES level $\alpha$ and a change in the nominal coverage probability $100(1-\gamma)\%$, respectively.

\begin{table}[h]
\centering
\resizebox{\textwidth}{!}{\begin{tabular}{rccccccc}
\hline \hline
\multicolumn{1}{c}{\begin{tabular}[c]{@{}c@{}}Sample\\ size\end{tabular}} &                         & \begin{tabular}[c]{@{}c@{}}Average\\ coverage\end{tabular} & \begin{tabular}[c]{@{}c@{}}Av. coverage\\ below/above\end{tabular} & \begin{tabular}[c]{@{}c@{}}Average\\ length\end{tabular} & \begin{tabular}[c]{@{}c@{}}Average\\ coverage\end{tabular} & \begin{tabular}[c]{@{}c@{}}Av. coverage\\ below/above\end{tabular} & \begin{tabular}[c]{@{}c@{}}Average\\ length\end{tabular} \\ \hline
\multicolumn{1}{c}{}                                                      &                         & \multicolumn{1}{l}{}                                       & \multicolumn{1}{l}{}                                               & \multicolumn{1}{l}{}                                     & \multicolumn{1}{l}{}                                       & \multicolumn{1}{l}{}                                               & \multicolumn{1}{l}{}                                     \\
\multicolumn{1}{c}{}                                                      &                         & \multicolumn{3}{c}{low persistence}                                                                                                                                                        & \multicolumn{3}{c}{high persistence}                                                                                                                                                       \\
$500$                                                                     & \multicolumn{1}{c:}{EP} & $86.45$                                                    & $4.65$/$8.90$                                                      & \multicolumn{1}{c:}{$0.446$}                             & $86.30$                                                    & $5.10$/$8.60$                                                      & $0.549$                                                  \\
                                                                          & \multicolumn{1}{c:}{RT} & $87.85$                                                    & $2.15$/$10.00$                                                     & \multicolumn{1}{c:}{$0.446$}                             & $87.65$                                                    & $2.40$/$9.95$                                                      & $0.549$                                                  \\
                                                                          & \multicolumn{1}{c:}{SY} & $88.10$                                                    & $2.80$/$9.10$                                                      & \multicolumn{1}{c:}{$0.455$}                             & $87.50$                                                    & $3.50$/$9.00$                                                      & $0.555$                                                  \\ \hdashline
$1,000$                                                                   & \multicolumn{1}{c:}{EP} & $89.05$                                                    & $4.60$/$6.35$                                                      & \multicolumn{1}{c:}{$0.309$}                             & $89.10$                                                    & $4.25$/$6.65$                                                      & $0.397$                                                  \\
                                                                          & \multicolumn{1}{c:}{RT} & $89.25$                                                    & $3.05$/$7.70$                                                      & \multicolumn{1}{c:}{$0.309$}                             & $88.90$                                                    & $3.10$/$8.00$                                                      & $0.397$                                                  \\
                                                                          & \multicolumn{1}{c:}{SY} & $89.65$                                                    & $3.60$/$6.75$                                                      & \multicolumn{1}{c:}{$0.311$}                             & $89.60$                                                    & $3.50$/$6.90$                                                      & $0.399$                                                  \\ \hdashline
$5,000$                                                                   & \multicolumn{1}{c:}{EP} & $90.15$                                                    & $4.35$/$5.50$                                                      & \multicolumn{1}{c:}{$0.145$}                             & $89.25$                                                    & $4.20$/$6.55$                                                      & $0.178$                                                  \\
                                                                          & \multicolumn{1}{c:}{RT} & $89.75$                                                    & $3.90$/$6.35$                                                      & \multicolumn{1}{c:}{$0.145$}                             & $88.90$                                                    & $3.95$/$7.15$                                                      & $0.178$                                                  \\
                                                                          & \multicolumn{1}{c:}{SY} & $90.15$                                                    & $4.10$/$5.75$                                                      & \multicolumn{1}{c:}{$0.145$}                             & $89.00$                                                    & $4.10$/$6.90$                                                      & $0.179$                                                  \\ \hdashline
$10,000$                                                                  & \multicolumn{1}{c:}{EP} & $89.60$                                                    & $4.70$/$5.70$                                                      & \multicolumn{1}{c:}{$0.103$}                             & $89.45$                                                    & $4.95$/$5.60$                                                      & $0.126$                                                  \\
                                                                          & \multicolumn{1}{c:}{RT} & $89.85$                                                    & $4.20$/$5.95$                                                      & \multicolumn{1}{c:}{$0.103$}                             & $89.70$                                                    & $4.45$/$5.85$                                                      & $0.126$                                                  \\
                                                                          & \multicolumn{1}{c:}{SY} & $89.95$                                                    & $4.30$/$5.75$                                                      & \multicolumn{1}{c:}{$0.103$}                             & $89.40$                                                    & $4.90$/$5.70$                                                      & $0.126$                                                  \\ \hline \hline
\end{tabular}}
\caption{The table reports distinct features of the fixed-design bootstrap confidence intervals for the conditional ES at \textbf{level} $\bm{\alpha=0.05}$ with \textbf{nominal coverage} $\bm{1-\gamma=90\%}$. For each interval type and different sample sizes ($n$), the interval's average coverage rates (in $\%$), the average rate of the conditional ES being below/above the interval (in $\%$) and the interval's average length are tabulated. The intervals are based on $B=2{,}000$ bootstrap replications and the averages are computed using $S=2{,}000$ simulations. The DGP is a GARCH(1,1) with \textbf{Gaussian innovations}.}
\label{tab:5.2}
\end{table}

Table \ref{tab:5.2} considers the case where the innovations follow a standard normal distribution. Results are qualitatively similar to the benchmark. In particular, coverage rates are generally close to the nominal level for $n\geq 1{,}000$ yet the under-coverage in smaller sample sizes is less in this scenario. 
For example, the average coverage is at most $3.70pp$ below the $90\%$ level even when $n=500$. In general, results seem to be ``less extreme" compared to the benchmark: $(i)$ the average length of all intervals is smaller for all sample sizes and $(ii)$ the average rate of the conditional ES being above the interval lies closer to the corresponding rate below the interval. Moreover, we observe that there is no interval that outperforms the others
in the case of $\eta_{t}$ being standard normally distributed.

\begin{table}[h]
\centering
\resizebox{\textwidth}{!}{\begin{tabular}{rccccccc}
\hline \hline
\multicolumn{1}{c}{\begin{tabular}[c]{@{}c@{}}Sample\\ size\end{tabular}} &                         & \begin{tabular}[c]{@{}c@{}}Average\\ coverage\end{tabular} & \begin{tabular}[c]{@{}c@{}}Av. coverage\\ below/above\end{tabular} & \begin{tabular}[c]{@{}c@{}}Average\\ length\end{tabular} & \begin{tabular}[c]{@{}c@{}}Average\\ coverage\end{tabular} & \begin{tabular}[c]{@{}c@{}}Av. coverage\\ below/above\end{tabular} & \begin{tabular}[c]{@{}c@{}}Average\\ length\end{tabular} \\ \hline
\multicolumn{1}{c}{}                                                      &                         &                                                            &                                                                    &                                                          &                                                            &                                                                    &                                                          \\
\multicolumn{1}{c}{}                                                      &                         & \multicolumn{3}{c}{low persistence}                                                                                                                                                        & \multicolumn{3}{c}{high persistence}                                                                                                                                                       \\
$500$                                                                     & \multicolumn{1}{c:}{EP} & $73.75$                                                    & $4.75$/$21.50$                                                     & \multicolumn{1}{c:}{$1.183$}                             & $72.25$                                                    & $5.45$/$22.30$                                                     & $1.566$                                                  \\
                                                                          & \multicolumn{1}{c:}{RT} & $71.30$                                                    & $0.75$/$27.95$                                                     & \multicolumn{1}{c:}{$1.183$}                             & $70.75$                                                    & $1.00$/$28.25$                                                     & $1.566$                                                  \\
                                                                          & \multicolumn{1}{c:}{SY} & $74.85$                                                    & $1.55$/$23.60$                                                     & \multicolumn{1}{c:}{$1.226$}                             & $73.70$                                                    & $2.00$/$24.30$                                                     & $1.616$                                                  \\ \hdashline
$1,000$                                                                   & \multicolumn{1}{c:}{EP} & $81.90$                                                    & $2.55$/$15.55$                                                     & \multicolumn{1}{c:}{$0.944$}                             & $81.55$                                                    & $2.80$/$15.65$                                                     & $1.260$                                                  \\
                                                                          & \multicolumn{1}{c:}{RT} & $80.55$                                                    & $0.85$/$18.60$                                                     & \multicolumn{1}{c:}{$0.944$}                             & $79.45$                                                    & $1.30$/$19.25$                                                     & $1.260$                                                  \\
                                                                          & \multicolumn{1}{c:}{SY} & $82.20$                                                    & $1.25$/$16.55$                                                     & \multicolumn{1}{c:}{$0.955$}                             & $81.20$                                                    & $1.55$/$17.25$                                                     & $1.270$                                                  \\ \hdashline
$5,000$                                                                   & \multicolumn{1}{c:}{EP} & $87.05$                                                    & $3.15$/$9.80$                                                      & \multicolumn{1}{c:}{$0.476$}                             & $88.30$                                                    & $2.30$/$9.40$                                                      & $0.632$                                                  \\
                                                                          & \multicolumn{1}{c:}{RT} & $87.35$                                                    & $2.55$/$10.10$                                                     & \multicolumn{1}{c:}{$0.476$}                             & $88.00$                                                    & $2.50$/$9.50$                                                      & $0.632$                                                  \\
                                                                          & \multicolumn{1}{c:}{SY} & $87.30$                                                    & $2.85$/$9.85$                                                      & \multicolumn{1}{c:}{$0.477$}                             & $88.25$                                                    & $2.40$/$9.35$                                                      & $0.632$                                                  \\ \hdashline
$10,000$                                                                  & \multicolumn{1}{c:}{EP} & $88.45$                                                    & $3.70$/$7.85$                                                      & \multicolumn{1}{c:}{$0.347$}                             & $89.10$                                                    & $3.00$/$7.90$                                                      & $0.458$                                                  \\
\multicolumn{1}{c}{}                                                      & \multicolumn{1}{c:}{RT} & $88.40$                                                    & $3.75$/$7.85$                                                      & \multicolumn{1}{c:}{$0.347$}                             & $88.45$                                                    & $3.40$/$8.15$                                                      & $0.458$                                                  \\
\multicolumn{1}{c}{}                                                      & \multicolumn{1}{c:}{SY} & $88.65$                                                    & $3.60$/$7.75$                                                      & \multicolumn{1}{c:}{$0.347$}                             & $88.55$                                                    & $3.20$/$8.25$                                                      & $0.458$                                                  \\ \hline \hline
\end{tabular}}
\caption{The table reports distinct features of the fixed-design bootstrap confidence intervals for the conditional ES at \textbf{level} $\bm{\alpha=0.01}$ with \textbf{nominal coverage} $\bm{1-\gamma=90\%}$. For each interval type and different sample sizes ($n$), the interval's average coverage rates (in $\%$), the average rate of the conditional ES being below/above the interval (in $\%$) and the interval's average length are tabulated. The intervals are based on $B=2{,}000$ bootstrap replications and the averages are computed using $S=2{,}000$ simulations. The DGP is a GARCH(1,1) with (normalized) \textbf{Student-$t$ innovations ($\bm{6}$ degrees of freedom)}.}
\label{tab:5.3}
\end{table}

Table \ref{tab:5.3} provides simulation results for the conditional ES at level $\alpha = 0.01$, where the DGP is a GARCH($1,1$) with Student-$t$ innovations (6 degrees of freedom). Unsurprisingly, we find that the average length of all intervals is considerably larger compared to the benchmark. More strikingly, we observe that the phenomenon of under-coverage appears across sample sizes. For the lowest sample size considered, i.e. $n=500$, average coverage rates are between $15pp$ and $20pp$ below nominal value. This problem is still severe for the case $n=1{,}000$, as rates are still approximately $10pp$ too low. Results are more satisfactory for the two highest sample sizes. An explanation for this result can be found in \citeauthor{gao2008estimation} (\citeyear{gao2008estimation}, Remark 3.3) who assert that the effective sample size for the estimation of ES is solely $n\alpha$.
All in all, we conclude that larger sample sizes are needed to obtain acceptable coverage probabilities.

\begin{table}[h]
\centering
\resizebox{\textwidth}{!}{\begin{tabular}{rccccccc}
\hline \hline
\multicolumn{1}{c}{\begin{tabular}[c]{@{}c@{}}Sample\\ size\end{tabular}} &                         & \begin{tabular}[c]{@{}c@{}}Average\\ coverage\end{tabular} & \begin{tabular}[c]{@{}c@{}}Av. coverage\\ below/above\end{tabular} & \begin{tabular}[c]{@{}c@{}}Average\\ length\end{tabular} & \begin{tabular}[c]{@{}c@{}}Average\\ coverage\end{tabular} & \begin{tabular}[c]{@{}c@{}}Av. coverage\\ below/above\end{tabular} & \begin{tabular}[c]{@{}c@{}}Average\\ length\end{tabular} \\ \hline
\multicolumn{1}{c}{}                                                      &                         &                                                            &                                                                    &                                                          &                                                            &                                                                    &                                                          \\
\multicolumn{1}{c}{}                                                      &                         & \multicolumn{3}{c}{low persistence}                                                                                                                                                        & \multicolumn{3}{c}{high persistence}                                                                                                                                                       \\
$500$                                                                     & \multicolumn{1}{c:}{EP} & $89.40$                                                    & $1.85$/$8.75$                                                      & \multicolumn{1}{c:}{$0.732$}                             & $89.65$                                                    & $2.10$/$8.25$                                                      & $0.984$                                                  \\
                                                                          & \multicolumn{1}{c:}{RT} & $90.65$                                                    & $1.00$/$8.35$                                                      & \multicolumn{1}{c:}{$0.732$}                             & $90.85$                                                    & $1.00$/$8.15$                                                      & $0.984$                                                  \\
                                                                          & \multicolumn{1}{c:}{SY} & $90.50$                                                    & $1.15$/$8.35$                                                      & \multicolumn{1}{c:}{$0.745$}                             & $91.10$                                                    & $1.30$/$7.60$                                                      & $0.997$                                                  \\ \hdashline
$1,000$                                                                   & \multicolumn{1}{c:}{EP} & $92.80$                                                    & $1.45$/$5.75$                                                      & \multicolumn{1}{c:}{$0.534$}                             & $93.05$                                                    & $1.15$/$5.80$                                                      & $0.727$                                                  \\
                                                                          & \multicolumn{1}{c:}{RT} & $93.50$                                                    & $1.15$/$5.35$                                                      & \multicolumn{1}{c:}{$0.534$}                             & $93.35$                                                    & $1.20$/$5.45$                                                      & $0.727$                                                  \\
                                                                          & \multicolumn{1}{c:}{SY} & $93.60$                                                    & $1.00$/$5.40$                                                      & \multicolumn{1}{c:}{$0.540$}                             & $93.40$                                                    & $1.15$/$5.45$                                                      & $0.730$                                                  \\ \hdashline
$5,000$                                                                   & \multicolumn{1}{c:}{EP} & $93.80$                                                    & $2.25$/$3.95$                                                      & \multicolumn{1}{c:}{$0.257$}                             & $94.35$                                                    & $1.60$/$4.05$                                                      & $0.342$                                                  \\
                                                                          & \multicolumn{1}{c:}{RT} & $94.45$                                                    & $2.05$/$3.50$                                                      & \multicolumn{1}{c:}{$0.257$}                             & $94.95$                                                    & $1.50$/$3.55$                                                      & $0.342$                                                  \\
                                                                          & \multicolumn{1}{c:}{SY} & $94.20$                                                    & $2.10$/$3.70$                                                      & \multicolumn{1}{c:}{$0.257$}                             & $94.60$                                                    & $1.55$/$3.85$                                                      & $0.342$                                                  \\ \hdashline
$10,000$                                                                  & \multicolumn{1}{c:}{EP} & $94.35$                                                    & $2.45$/$3.20$                                                      & \multicolumn{1}{c:}{$0.185$}                             & $94.25$                                                    & $1.85$/$3.90$                                                      & $0.243$                                                  \\
\multicolumn{1}{c}{}                                                      & \multicolumn{1}{c:}{RT} & $95.25$                                                    & $2.35$/$2.40$                                                      & \multicolumn{1}{c:}{$0.185$}                             & $94.50$                                                    & $2.05$/$3.45$                                                      & $0.243$                                                  \\
\multicolumn{1}{c}{}                                                      & \multicolumn{1}{c:}{SY} & $94.90$                                                    & $2.50$/$2.60$                                                      & \multicolumn{1}{c:}{$0.185$}                             & $94.70$                                                    & $1.85$/$3.45$                                                      & $0.243$                                                  \\ \hline \hline
\end{tabular}}
\caption{The table reports distinct features of the fixed-design bootstrap confidence intervals for the conditional ES at \textbf{level} $\bm{\alpha=0.05}$ with \textbf{nominal coverage} $\bm{1-\gamma=95\%}$. For each interval type and different sample sizes ($n$), the interval's average coverage rates (in $\%$), the average rate of the conditional ES being below/above the interval (in $\%$) and the interval's average length are tabulated. The intervals are based on $B=2{,}000$ bootstrap replications and the averages are computed using $S=2{,}000$ simulations. The DGP is a GARCH(1,1) with (normalized) \textbf{Student-$t$ innovations ($\bm{6}$ degrees of freedom)}.}
\label{tab:5.4}
\end{table}

Table \ref{tab:5.4} considers an increase in the interval's nominal value from $90\%$ to $95\%$. Once again we conclude that the results are qualitatively similar to the benchmark. 
The average lengths of the intervals are larger for every sample size considered. These results are to be expected for bootstrap intervals with higher nominal value.

It might be of interest to compare the results for the conditional ES with those reported in \cite{beutner2018residual} for the conditional VaR. They find that the EP interval performs worse than the RT interval in small samples, which is in line with the theoretical findings in \cite{falk1991coverage}. This result does not carry over to the conditional ES, since in most instances (except Table \ref{tab:5.4}) the EP interval even outperforms the RT interval. To make a full comparison, we also computed the results where the DGP is a T-GARCH($1,1$). Results are vastly comparable and available upon request.


To summarize, the simulation study suggests that the fixed-design bootstrap works well in terms of average coverage. In comparison to the conditional VaR, higher sample sizes are necessary to obtain coverage rates close to the nominal value. There is no clear evidence to have a preference for any of the three intervals based on the simulation results in all different settings. This directly contrasts results for the conditional VaR in \cite{beutner2018residual} for which the RT bootstrap interval is found superior.

\section{Conclusion}
\label{sec:5.6}

This paper studies the two-step estimation procedure of \cite{francq2015risk} associated with the conditional ES. In the first step, the conditional volatility parameters are estimated by QMLE, while the second step corresponds to the estimation of conditional ES based on the first-step residuals. We find that the estimators of the parameters that comprise the conditional ES have a joint asymptotic distribution that does not depend on any density. This is in direct contrast with the conditional VaR estimator for which density estimation is required. A fixed-design residual bootstrap method is proposed to mimic the finite sample distribution of the two-step estimator and its consistency is proven under mild assumptions. In addition, an algorithm is provided for the construction of bootstrap intervals for the conditional ES to take into account the uncertainty induced by estimation. Three interval types are suggested and a simulation study is conducted to investigate their performance in finite samples. Firstly, we find that average coverage rates of all intervals are close to nominal value, except when sample size is low. Secondly, we find that there is no clear evidence that any of the proposed intervals outperforms the others. This contrasts the results in \cite{beutner2018residual} for the conditional VaR, who find superiority of the reversed-tails bootstrap interval.

The present work can be extended by developing a bootstrap procedure for the one-step approach estimator of \cite{francq2015risk}. This suggestion is left for future research.

\newpage

\appendix

\section{Appendix}
\subsection{Auxiliary Results and Proofs}
\label{app:5.A}

\begin{lemma}
\label{lem:5.1}
Suppose Assumptions \ref{as:5.1}--\ref{as:5.10} hold with $a=\pm 12$, $b=12$ and $c=6$. Then, we have $A_n^*\overset{p^*}{\to}0$ in probability.
\end{lemma}

\begin{proof}
The proof is inspired by \citeauthor{chen2007nonparametric} (\citeyear{chen2007nonparametric}, proof of Lemma 2). Take $\delta \in (1/4,1/2)$, and expand 
\begin{align*}
A_n^* =& \frac{1}{\sqrt{n}}\sum_{t=1}^n\big(\hat{\eta}_t^* -\hat{\xi}_{n,\alpha}\big)\big(\mathbbm{1}_{\{\hat{\eta}_t^*< \hat{\xi}_{n,\alpha}^*\}}-\mathbbm{1}_{\{\eta_t^*< \hat{\xi}_{n,\alpha}\}}\big)\\
=& \frac{1}{\sqrt{n}}\sum_{t=1}^n\big(\hat{\eta}_t^* -\hat{\xi}_{n,\alpha}\big)\Big(\mathbbm{1}_{\{\hat{\xi}_{n,\alpha}\leq \eta_t^* \}}\mathbbm{1}_{\{\hat{\eta}_t^*<\hat{\xi}_{n,\alpha}^* \}}-\mathbbm{1}_{\{\hat{\xi}_{n,\alpha}^*\leq \hat{\eta}_t^*\}}\mathbbm{1}_{\{\eta_t^*<\hat{\xi}_{n,\alpha} \}}\Big)\\ 
=& \underbrace{\frac{1}{\sqrt{n}}\sum_{t=1}^n\big(\hat{\eta}_t^* -\hat{\xi}_{n,\alpha}\big)\Big(\mathbbm{1}_{\{\hat{\xi}_{n,\alpha}\leq \eta_t^* \}}\mathbbm{1}_{\{\hat{\eta}_t^*<\hat{\xi}_{n,\alpha}^* \}}-\mathbbm{1}_{\{\hat{\xi}_{n,\alpha}^*\leq \hat{\eta}_t^*\}}\mathbbm{1}_{\{\eta_t^*<\hat{\xi}_{n,\alpha} \}}\Big)\mathbbm{1}_{\{ |\hat{\eta}_t^*-\eta_t^*|+|\hat{\xi}_{n,\alpha}^*-\hat{\xi}_{n,\alpha}|< n^{-\delta}\}}}_{I}\\
&+\underbrace{\frac{1}{\sqrt{n}}\sum_{t=1}^n\big(\hat{\eta}_t^* -\hat{\xi}_{n,\alpha}\big)\Big(\mathbbm{1}_{\{\hat{\xi}_{n,\alpha}\leq \eta_t^* \}}\mathbbm{1}_{\{\hat{\eta}_t^*<\hat{\xi}_{n,\alpha}^* \}}-\mathbbm{1}_{\{\hat{\xi}_{n,\alpha}^*\leq \hat{\eta}_t^*\}}\mathbbm{1}_{\{\eta_t^*<\hat{\xi}_{n,\alpha} \}}\Big)\mathbbm{1}_{\{ |\hat{\eta}_t^*-\eta_t^*|+|\hat{\xi}_{n,\alpha}^*-\hat{\xi}_{n,\alpha}|\geq n^{-\delta}\}}.}_{II}
\end{align*}
The first term can be bounded by
\begin{align*}
|I| &\leq\! \frac{1}{\sqrt{n}}\sum_{t=1}^n\big|\hat{\eta}_t^* -\hat{\xi}_{n,\alpha}\big|\Big(\mathbbm{1}_{\{\hat{\xi}_{n,\alpha}\leq \eta_t^* \}}\mathbbm{1}_{\{\hat{\eta}_t^*<\hat{\xi}_{n,\alpha}^* \}}+\mathbbm{1}_{\{\hat{\xi}_{n,\alpha}^*\leq \hat{\eta}_t^*\}}\mathbbm{1}_{\{\eta_t^*<\hat{\xi}_{n,\alpha} \}}\Big)\mathbbm{1}_{\{ |\hat{\eta}_t^*-\eta_t^*|+|\hat{\xi}_{n,\alpha}^*-\hat{\xi}_{n,\alpha}|< n^{-\delta}\}}\\
&\leq \frac{1}{\sqrt{n}}\sum_{t=1}^n\big(n^{-\delta}+|\eta_t^* -\hat{\xi}_{n,\alpha}|\big)\Big(\underbrace{\mathbbm{1}_{\{\hat{\xi}_{n,\alpha} \leq \eta_t^*\}} \mathbbm{1}_{\{\eta_t^*-n^{-\delta}<\hat{\xi}_{n,\alpha}+n^{-\delta} \}}}_{=\mathbbm{1}_{\{0\leq \eta_t^*-\hat{\xi}_{n,\alpha}< 2n^{-\delta} \}}}+\underbrace{\mathbbm{1}_{\{\hat{\xi}_{n,\alpha}-n^{-\delta} \leq \eta_t^*+n^{-\delta}\}}\mathbbm{1}_{\{\eta_t^*<\hat{\xi}_{n,\alpha} \}}}_{=\mathbbm{1}_{\{-2 n^{-\delta}\leq\eta_t^*-\hat{\xi}_{n,\alpha}<0 \}}}\Big)\\
&\leq \frac{2}{\sqrt{n}}\sum_{t=1}^n\big(n^{-\delta}+|\eta_t^* -\hat{\xi}_{n,\alpha}|\big)\mathbbm{1}_{\{|\eta_t^*-\hat{\xi}_{n,\alpha}|\leq 2 n^{-\delta} \}}  \leq   \underbrace{\frac{6}{\sqrt{n}n^{\delta}}\sum_{t=1}^n\mathbbm{1}_{\{|\eta_t^*-\hat{\xi}_{n,\alpha}|\leq 2 n^{-\delta} \}}}_{W_n^*}
\end{align*}
Note that, given the original sample, the random variables $X_{t,n}^* = \mathbbm{1}_{\{ |\eta_t^*-\hat{\xi}_{n,\alpha}|< 2 n^{-\delta}\}}$ are \textit{i.i.d.} The conditional mean satisfies $\EE^*[X_{t,n}^*]=O_p(n^{-\delta})$; to appreciate why, we have
\begin{align*}
\EE^*[X_{t,n}^*] &\leq \PP^*\big[\hat{\xi}_{n,\alpha}-2n^{-\delta} <\eta_t^*\leq  \hat{\xi}_{n,\alpha}+2 n^{-\delta}\big]\\
&= \hat{\mathbbm{F}}_n\big(\hat{\xi}_{n,\alpha}+2n^{-\delta}\big) - \hat{\mathbbm{F}}_n\big(\hat{\xi}_{n,\alpha}-2n^{-\delta}\big)\\
&= \bigg(\mathbbm{F}_n\big(\hat{\xi}_{n,\alpha}+ 2n^{-\delta}\big)+\big(\hat{\xi}_{n,\alpha}+ 2n^{-\delta}\big)f\big(\hat{\xi}_{n,\alpha}+ 2n^{-\delta}\big)\Omega'\big(\hat{\theta}_n-\theta_0\big)+o_p(n^{-1/2})\bigg)\\
& \ \ \ - \bigg(\mathbbm{F}_n\big(\hat{\xi}_{n,\alpha}- 2n^{-\delta}\big)+\big(\hat{\xi}_{n,\alpha}- 2n^{-\delta}\big)f\big(\hat{\xi}_{n,\alpha}+ 2n^{-\delta}\big)\Omega'\big(\hat{\theta}_n-\theta_0\big)+o_p(n^{-1/2})\bigg)\\
&= \mathbbm{F}_n\big(\hat{\xi}_{n,\alpha}+ 2n^{-\delta}\big)- \mathbbm{F}_n\big(\hat{\xi}_{n,\alpha}- 2n^{-\delta}\big)+o_p(n^{-1/2})\\
&= F\big(\hat{\xi}_{n,\alpha}+ 2n^{-\delta}\big)-F\big(\hat{\xi}_{n,\alpha}- 2n^{-\delta}\big)+o_p(n^{-1/2})\\
&= 4n^{-\delta}f\big(\hat{\xi}_{n,\alpha}+ b_n\big)+o_p(n^{-1/2})=4n^{-\delta}f(\xi_\alpha)+o_p(n^{-\delta})+o_p(n^{-1/2})
\end{align*}
where $b_n \in (-2n^{-\delta},2n^{-\delta})$. The second equality follows from  \citeauthor{beutner2018residual} (\citeyear{beutner2018residual}, \textit{Step 5} of the proof of Lemma 3) whereas the third equality is implied by $\hat{\xi}_{n,\alpha}+2n^{-\delta}\overset{a.s}{\to}\xi_\alpha$, continuity of $f$ in a neighborhood of $\xi_\alpha$ and $\sqrt{n}\big(\hat{\theta}_n-\theta_0\big)=O_p(1)$. The fourth equality is due to \textit{Step 6} in the proof of \citeauthor{beutner2018residual} (\citeyear{beutner2018residual}, Lemma 3). The mean value theorem is applied to obtain the fifth equality and the last equality is due to continuity of $f$ in a neighborhood of $\xi_\alpha$ and $\hat{\xi}_{n,\alpha}+b_n\overset{a.s}{\to}\xi_\alpha$. Thus, we have
\begin{align*}
    \EE^*\big[W_n^*\big] =& \frac{6}{\sqrt{n}n^\delta}\sum_{t=1}^n \EE^*\big[X_{t,n}^*\big]=O_p(\sqrt{n}/n^{2\delta})
    = O_p(n^{-\delta}) \\
    \Var^*\big[W_n^*\big] =& \frac{36}{n^{1+2\delta}}\sum_{t=1}^n \EE^*\big[X_{t,n}^*\big]\Big(\underbrace{1-\EE^*\big[X_{t,n}^*\big]}_{\leq 1}\Big)=O_p(n^{-3\delta})
\end{align*}
implying $W_n^*\overset{p^*}{\to}0$ in probability and we conclude that $I\overset{p^*}{\to}0$ in probability. Regarding the second term, we write $\hat{\mathbbm{1}}_{n,t}^*=\mathbbm{1}_{\{ |\hat{\eta}_t^*-\eta_t^*|+|\hat{\xi}_{n,\alpha}^*-\hat{\xi}_{n,\alpha}|\geq n^{-\delta}\}}$ and establish the following bound
\begin{align*}
|II| &\leq  \frac{1}{\sqrt{n}}\sum_{t=1}^n\big|\hat{\eta}_t^* -\hat{\xi}_{n,\alpha}\big|\Big(\mathbbm{1}_{\{\hat{\xi}_{n,\alpha}\leq \eta_t^* \}}\mathbbm{1}_{\{\hat{\eta}_t^*<\hat{\xi}_{n,\alpha}^* \}}+\mathbbm{1}_{\{\hat{\xi}_{n,\alpha}^*\leq \hat{\eta}_t^*\}}\mathbbm{1}_{\{\eta_t^*<\hat{\xi}_{n,\alpha} \}}\Big)\hat{\mathbbm{1}}_{n,t}^*\\
&= \frac{1}{\sqrt{n}}\sum_{t=1}^n\big|\hat{\eta}_t^*-\hat{\xi}_{n,\alpha}\big|\Big(\mathbbm{1}_{\{\hat{\xi}_{n,\alpha} \leq \eta_t^* \}}\mathbbm{1}_{\{\hat{\eta}_t^*<\hat{\xi}_{n,\alpha}^* \}}\mathbbm{1}_{\{\hat{\eta}_t^* < \hat{\xi}_{n,\alpha} \}}+\mathbbm{1}_{\{ \hat{\xi}_{n,\alpha}^* \leq \hat{\eta}_t^*\}}\mathbbm{1}_{\{\eta_t^*<\hat{\xi}_{n,\alpha} \}}\mathbbm{1}_{\{\hat{\xi}_{n,\alpha} \leq \hat{\eta}_t^*\}}\Big)\hat{\mathbbm{1}}_{n,t}^*\\
&\ \ \ +\frac{1}{\sqrt{n}}\sum_{t=1}^n\big|\hat{\eta}_t^* -\hat{\xi}_{n,\alpha}\big|\Big(\mathbbm{1}_{\{\hat{\xi}_{n,\alpha} \leq \eta_t^* \}}\mathbbm{1}_{\{\hat{\eta}_t^*<\hat{\xi}_{n,\alpha}^* \}}\mathbbm{1}_{\{\hat{\xi}_{n,\alpha} \leq \hat{\eta}_t^*  \}}+\mathbbm{1}_{\{\hat{\xi}_{n,\alpha}^* \leq \hat{\eta}_t^*\}}\mathbbm{1}_{\{\eta_t^*<\hat{\xi}_{n,\alpha} \}}\mathbbm{1}_{\{\hat{\eta}_t^* < \hat{\xi}_{n,\alpha} \}}\Big)\hat{\mathbbm{1}}_{n,t}^* \\
&\leq \frac{1}{\sqrt{n}}\sum_{t=1}^n2\big|\hat{\eta}_t^*-\eta_t^*\big|\hat{\mathbbm{1}}_{n,t}^*+\frac{1}{\sqrt{n}}\sum_{t=1}^n 2\big|\hat{\xi}_{n,\alpha}^* -\hat{\xi}_{n,\alpha}\big|\hat{\mathbbm{1}}_{n,t}^*.
\end{align*}
The Taylor expansion in \citeauthor{beutner2018residual} (\citeyear{beutner2018residual}, Equation A.69) and $\hat{\eta}_t^*=\frac{\tilde{\sigma}_t(\hat{\theta}_n)}{\tilde{\sigma}_t(\hat{\theta}_n^*)}\eta_t^*$ gives
\begin{align}
\label{eq:5.A.1}
\begin{split}
\eta_t^*-\hat{\eta}_t^* =& \hat{D}_t'\big(\hat{\theta}_n^*-\hat{\theta}_n\big)\eta_t^*\\
&+ \frac{1}{2}\big(\hat{\theta}_n^*-\hat{\theta}_n\big)'\frac{\tilde{\sigma}_t(\hat{\theta}_n)}{\tilde{\sigma}_t(\breve{\theta}_n)}\Big(\tilde{H}_t(\breve{\theta}_n)-2\tilde{D}_t(\breve{\theta}_n)\tilde{D}_t'(\breve{\theta}_n)\Big) \big(\hat{\theta}_n^*-\hat{\theta}_n\big)\eta_t^*
\end{split}
\end{align}
with $\breve{\theta}_n$ between $\hat{\theta}_n^*$ and $\hat{\theta}_n$. Employing the Cauchy-Schwarz inequality we obtain 
\begin{equation*}
\big|\hat{\eta}_t^*-\eta_t^*\big| = \big|\big|\hat{\theta}_n^*-\hat{\theta}_n\big|\big|\: \big|\big|\hat{D}_t\big|\big|\:|\eta_t^*|+\frac{1}{2}\big|\big|\hat{\theta}_n^*-\hat{\theta}_n\big|\big|^2\frac{\tilde{\sigma}_t(\hat{\theta}_n)}{\tilde{\sigma}_t(\breve{\theta}_n)}\Big(\big|\big|\tilde{H}_t(\breve{\theta}_n)\big|\big|+2\big|\big|\tilde{D}_t(\breve{\theta}_n)\big|\big|^2\Big)\: |\eta_t^*|.
\end{equation*}
Hence, $|II|$ can further be bounded by
\begin{align*}
|II| &\leq n\big|\big|\hat{\theta}_n^*-\hat{\theta}_n\big|\big|^2\underbrace{\frac{1}{\sqrt{n}n}\sum_{t=1}^n \frac{\tilde{\sigma}_t(\hat{\theta}_n)}{\tilde{\sigma}_t(\breve{\theta}_n)}\Big(\big|\big|\tilde{H}_t(\breve{\theta}_n)\big|\big|+2\big|\big|\tilde{D}_t(\breve{\theta}_n)\big|\big|^2\Big) |\eta_t^*|}_{II_1}\\
& \ \ \ \ + 2\sqrt{n}\big|\hat{\xi}_{n,\alpha}^* -\hat{\xi}_{n,\alpha}\big|\underbrace{\frac{1}{n}\sum_{t=1}^n \hat{\mathbbm{1}}_{n,t}^*}_{II_2} + 2 \sqrt{n}\big|\big|\hat{\theta}_n^*-\hat{\theta}_n\big|\big|\underbrace{\frac{1}{n}\sum_{t=1}^n \big|\big|\hat{D}_t\big|\big|\:|\eta_t^*| \hat{\mathbbm{1}}_{n,t}^*}_{II_3}.
\end{align*}
Since $\sqrt{n}\big(\hat{\theta}_n^*-\hat{\theta}_n\big)\overset{d^*}{\to}N\big(0,\frac{\kappa-1}{4}J^{-1}\big)$ almost surely and  $\sqrt{n}\big(\hat{\xi}_{n,\alpha}^* -\hat{\xi}_{n,\alpha}\big)\overset{d^*}{\to}N(0,\zeta_\alpha)$ in probability for some $\zeta_\alpha$ (see \citeauthor{beutner2018residual}, \citeyear{beutner2018residual}, Proposition 1 and Theorem 3), it remains to show that the under-braced terms converge in conditional probability to zero in probability. Consider $II_1$; for every $\varepsilon>0$ we obtain 
\begin{align} \label{eq:5.A.2}
&\PP^*\big[II_1\geq \varepsilon \big] \leq  \PP^*\Big[II_1\geq \varepsilon \cap \breve{\theta}_n \in \mathscr{V}(\theta_0)\Big] + \PP^*\Big[\breve{\theta}_n \notin \mathscr{V}(\theta_0)\Big] \notag \\
&\leq \PP^*\Bigg[\frac{1}{\sqrt{n}n}\sum_{t=1}^n\underbrace{\sup_{\theta \in \mathscr{V}(\theta_0)}\frac{\tilde{\sigma}_t(\hat{\theta}_n)}{\tilde{\sigma}_t(\theta)}\bigg(\sup_{\theta \in \mathscr{V}(\theta_0)}\big|\big|\tilde{H}_t(\theta)\big|\big|+2\sup_{\theta \in \mathscr{V}(\theta_0)}\big|\big|\tilde{D}_t(\theta)\big|\big|^2\bigg)}_{Q_{t,n}} |\eta_t^*|\geq \varepsilon\Bigg] \notag \\
& \ \ \ + \PP^*\Big[\breve{\theta}_n \notin \mathscr{V}(\theta_0)\Big] \\
\nonumber
& \leq \frac{1}{\varepsilon} \EE^*\bigg[\frac{1}{\sqrt{n}n}\sum_{t=1}^n Q_{t,n} |\eta_t^*|\bigg] + \PP^*\Big[\breve{\theta}_n \notin \mathscr{V}(\theta_0)\Big] 
= \frac{\EE^*\big[|\eta_t^*|\big]}{\varepsilon \sqrt{n}n}\sum_{t=1}^n Q_{t,n} + \PP^*\Big[\breve{\theta}_n \notin \mathscr{V}(\theta_0)\Big] \notag,
\end{align}
where the last inequality is due to Markov. As $\EE^*\big[|\eta_t^*|\big]\leq \big(\EE^*\big[\eta_t^{*2}\big]\big)^{1/2}\overset{a.s.}{\to}\big(\EE\big[\eta_t^{2}\big]\big)^{1/2}=1$, $\frac{1}{n}\sum_{t=1}^n Q_{t,n}=O_p(1)$ and $\PP^*\big[\breve{\theta}_n \notin \mathscr{V}(\theta_0)\big]\overset{a.s.}{\to}0$ by \citeauthor{beutner2018residual} (\citeyear{beutner2018residual}, Lemma 4 and proof of Lemma 10), we have $\PP^*\big[II_1\geq \varepsilon \big]\overset{p}{\to}0$ and we conclude that $II_1\overset{p^*}{\to}0$ in probability. Next, we focus on $II_2$. For every $\varepsilon>0$ we obtain 
{\allowdisplaybreaks
\begin{align*}
&\PP^*\big[II_2\geq \varepsilon\big]
\leq \PP^*\bigg[\frac{1}{n}\sum_{t=1}^n \mathbbm{1}_{\{ |\hat{\eta}_t^*-\eta_t^*|+|\hat{\xi}_{n,\alpha}^*-\hat{\xi}_{n,\alpha}|\geq n^{-\delta}\}}\geq \varepsilon  \cap \breve{\theta}_n \in \mathscr{V}(\theta_0)\bigg]+ \PP^*\Big[\breve{\theta}_n \notin \mathscr{V}(\theta_0)\Big]\\
&\leq  \PP^*\bigg[\frac{1}{n}\sum_{t=1}^n \mathbbm{1}_{\{||\hat{\theta}_n^*-\hat{\theta}_n||\:||\hat{D}_t||\:|\eta_t^*|+\frac{1}{2}||\hat{\theta}_n^*-\hat{\theta}_n||^2Q_{n,t}\: |\eta_t^*|+|\hat{\xi}_{n,\alpha}^*-\hat{\xi}_{n,\alpha}|\geq n^{-\delta}\}}\geq \varepsilon  \bigg]+ \PP^*\Big[\breve{\theta}_n \notin \mathscr{V}(\theta_0)\Big] \\
&\leq  \PP^*\bigg[\frac{1}{n}\sum_{t=1}^n \mathbbm{1}_{\{\frac{C}{\sqrt{n}}\:||\hat{D}_t||\:|\eta_t^*|+\frac{C^2}{2n}Q_{n,t}\: |\eta_t^*|+\frac{C}{\sqrt{n}}\geq n^{-\delta}\}}\geq \varepsilon  \bigg] \\
&\quad +\PP^*\Big[ \sqrt{n}\big|\hat{\xi}_{n,\alpha}^* -\hat{\xi}_{n,\alpha}\big|> C \Big]+ \PP^*\Big[\sqrt{n}\big|\big|\hat{\theta}_n^*-\hat{\theta}_n\big|\big|> C \Big]+ \PP^*\Big[\breve{\theta}_n \notin \mathscr{V}(\theta_0)\Big]\\
&\leq  \frac{1}{\varepsilon n}\sum_{t=1}^n \PP^*\bigg[\frac{C}{\sqrt{n}}\:||\hat{D}_t||\:|\eta_t^*|+\frac{C^2}{2n}Q_{n,t}\: |\eta_t^*|+\frac{C}{\sqrt{n}}\geq n^{-\delta}  \bigg]\\
&\quad +\PP^*\Big[ \sqrt{n}\big|\hat{\xi}_{n,\alpha}^* -\hat{\xi}_{n,\alpha}\big|> C \Big]+ \PP^*\Big[\sqrt{n}\big|\big|\hat{\theta}_n^*-\hat{\theta}_n\big|\big|> C \Big]+ \PP^*\Big[\breve{\theta}_n \notin \mathscr{V}(\theta_0)\Big]\\
&\leq  \frac{n^\delta}{\varepsilon n }\sum_{t=1}^n \EE^*\bigg[\frac{C}{\sqrt{n}}\:||\hat{D}_t||\:|\eta_t^*|+\frac{C^2}{2n}Q_{n,t}\: |\eta_t^*|+\frac{C}{\sqrt{n}}\bigg]\\
&\quad +\PP^*\Big[ \sqrt{n}\big|\hat{\xi}_{n,\alpha}^* -\hat{\xi}_{n,\alpha}\big|> C \Big]+ \PP^*\Big[\sqrt{n}\big|\big|\hat{\theta}_n^*-\hat{\theta}_n\big|\big|> C \Big]+ \PP^*\Big[\breve{\theta}_n \notin \mathscr{V}(\theta_0)\Big]\\
&=  \underbrace{\frac{n^\delta}{\sqrt{n}} \frac{C \EE^*\big[|\eta_t^*|\big]}{\varepsilon}  \frac{1}{n}\sum_{t=1}^n ||\hat{D}_t||}_{II_{2,1}}+\underbrace{\frac{n^\delta}{n} \frac{C^2 \EE^*\big[|\eta_t^*|\big]}{2 \varepsilon}\frac{1}{n}\sum_{t=1}^nQ_{n,t}}_{II_{2,2}} +\underbrace{\frac{n^\delta}{\sqrt{n}}\frac{C}{\varepsilon}}_{II_{2,3}}\\
&\quad +\underbrace{\PP^*\Big[ \sqrt{n}\big|\hat{\xi}_{n,\alpha}^* -\hat{\xi}_{n,\alpha}\big|> C \Big]}_{II_{2,4}}+ \underbrace{\PP^*\Big[\sqrt{n}\big|\big|\hat{\theta}_n^*-\hat{\theta}_n\big|\big|> C \Big]}_{II_{2,5}}+ \underbrace{\PP^*\Big[\breve{\theta}_n \notin \mathscr{V}(\theta_0)\Big]}_{II_{2,6}}.
\end{align*}}
Previously, we have shown that $II_{2,6}\overset{a.s.}{\to}0$, whereas the $II_{2,4}$ and $II_{2,4}$ can be made arbitrarily small in probability by choosing $C$ sufficiently large. Given $C$, we find $II_{2,3}\to 0$ as $\frac{n^\delta}{\sqrt{n}}\to 0$. Recalling that $\lim_{n \to \infty}\EE^*\big[|\eta_t^*|\big]\leq 1$ almost surely and  $\sum_{t=1}^nQ_{n,t}=O_p(1)$ we find $II_{2,2}\overset{p}{\to}0$. The Cauchy-Schwarz inequality gives $\frac{1}{n} \sum_{t=1}^n ||\hat{D}_t||\leq \big(\frac{1}{n} \sum_{t=1}^n \big|\big|\hat{D}_t\big|\big|^4\big)^{1/4}$ and together with $\lim_{n \to \infty}\frac{1}{n} \sum_{t=1}^n \big|\big|\hat{D}_t\big|\big|^4<\infty$ almost surely (see \citeauthor{beutner2018residual}, \citeyear{beutner2018residual}, proof of Lemma 7), we get that $II_{2,1}\overset{a.s.}{\to}0$. Thus, $\PP^*\big[II_2\geq \varepsilon\big]\overset{p}{\to}0$ and we establish $II_2\overset{p^*}{\to}0$ in probability. Regarding $II_3$, H\"older's inequality implies
\begin{align*}
    II_3 \leq \bigg(\frac{1}{n}\sum_{t=1}^n \eta_t^{*2} \bigg)^{\frac{1}{2}}
    \bigg(\frac{1}{n}\sum_{t=1}^n \big|\big|\hat{D}_t\big|\big|^4\bigg)^{\frac{1}{4}}
    \bigg(\underbrace{\frac{1}{n}\sum_{t=1}^n  \hat{\mathbbm{1}}_{n,t}^*}_{=II_2}\bigg)^{\frac{1}{4}}.
\end{align*}
We have $\frac{1}{n}\sum_{t=1}^n |\eta_t^*|^2\overset{p^*}{\to}\EE[\eta_t^2]$ almost surely as
\begin{align*}
    \EE^*\bigg[\frac{1}{n}\sum_{t=1}^n \eta_t^{*2}\bigg] =   \EE^*\big[\eta_t^{*2}\big]\overset{a.s.}{\to}\EE[\eta_t^2] \quad \text{and} \quad 
    \Var^*\bigg[\frac{1}{n}\sum_{t=1}^n \eta_t^{*2}\bigg] =   \frac{1}{n}\Var^*\big[\eta_t^{*2}\big]\overset{a.s.}{\to}0
\end{align*}
by \citeauthor{beutner2018residual} (\citeyear{beutner2018residual}, Lemma 4). Recalling that 
$\lim_{n \to \infty}\frac{1}{n} \sum_{t=1}^n \big|\big|\hat{D}_t\big|\big|^4<\infty$ almost surely and $II_2\overset{p^*}{\to}0$ in probability we establish $II_3\overset{p^*}{\to}0$ in probability.
\end{proof}

\begin{lemma}
\label{lem:5.2}
Suppose Assumptions \ref{as:5.1}--\ref{as:5.4}, \ref{as:5.5}(\ref{as:5.5.1}), \ref{as:5.5}(\ref{as:5.5.3}), \ref{as:5.6}, \ref{as:5.7}, \ref{as:5.9} and \ref{as:5.10} hold with $a=\pm 12$, $b=12$ and $c=6$. Then, we have $C_n^* = \alpha \mu_\alpha \Omega \sqrt{n}\big(\hat{\theta}_n^*-\hat{\theta}_n\big)+o_{p^*}(1)$ in probability.
\end{lemma}
\begin{proof}
Inserting \eqref{eq:5.A.1} into the definition of $C_n^*$ we obtain
\begin{align*}
C_n^* &=\underbrace{-\frac{1}{n}\sum_{t=1}^n\hat{D}_t'\eta_t^*\mathbbm{1}_{\{\eta_t^*< \hat{\xi}_{n,\alpha}\}}}_{I}\sqrt{n}\big(\hat{\theta}_n^*-\hat{\theta}_n\big)\\
& \ \ \ \ -\frac{1}{2}\sqrt{n}\big(\hat{\theta}_n^*-\hat{\theta}_n\big)'\underbrace{\frac{1}{\sqrt{n}n}\sum_{t=1}^n\frac{\tilde{\sigma}_t(\hat{\theta}_n)}{\tilde{\sigma}_t(\breve{\theta}_n)}\Big(\tilde{H}_t(\breve{\theta}_n)-2\tilde{D}_t(\breve{\theta}_n)\tilde{D}_t'(\breve{\theta}_n)\Big) \eta_t^*}_{II}\mathbbm{1}_{\{\eta_t^*< \hat{\xi}_{n,\alpha}\}}\sqrt{n}\big(\hat{\theta}_n^*-\hat{\theta}_n\big).
\end{align*}
Since $\sqrt{n}\big(\hat{\theta}_n^*-\hat{\theta}_n\big)\overset{d^*}{\to}N\big(0,\frac{\kappa-1}{4}J^{-1}\big)$ almost surely (\citeauthor{beutner2018residual}, \citeyear{beutner2018residual}, Proposition 1) and $II\overset{p^*}{\to}0$ in probability by \eqref{eq:5.A.2}, it remains to show that $I\overset{p^*}{\to}\alpha \mu_\alpha \Omega'$ in probability. Noting that $-\EE\big[\eta_t\mathbbm{1}_{\{\eta_t< \xi_\alpha\}}\big]=\alpha \mu_\alpha$, we obtain
\begin{align*}
\EE^*[I]=&-\EE^*\big[\eta_t^*\mathbbm{1}_{\{\eta_t^*< \hat{\xi}_{n,\alpha}\}}\big]\frac{1}{n}\sum_{t=1}^n\hat{D}_t'=-\EE^*\big[\eta_t^*\mathbbm{1}_{\{\eta_t^*< \hat{\xi}_{n,\alpha}\}}\big]\hat{\Omega}_n'\overset{a.s.}{\to} \alpha \mu_\alpha \Omega'\\
\Var^*[I]=&\Var^*\big[\eta_t^*\mathbbm{1}_{\{\eta_t^*< \hat{\xi}_{n,\alpha}\}}\big]\frac{1}{n^2}\sum_{t=1}^n\hat{D}_t\hat{D}_t'= \frac{1}{n} \Var^*\big[\eta_t^*\mathbbm{1}_{\{\eta_t^*< \hat{\xi}_{n,\alpha}\}}\big]\hat{J}_n \overset{a.s.}{\to}0
\end{align*}
by \citeauthor{beutner2018residual} (\citeyear{beutner2018residual}, Lemma 2 and 4), which completes the proof.
\end{proof}

\begin{lemma}
\label{lem:5.3}
Suppose Assumptions \ref{as:5.1}--\ref{as:5.4}, \ref{as:5.5}(\ref{as:5.5.1}), \ref{as:5.5}(\ref{as:5.5.3}), \ref{as:5.6}, \ref{as:5.9} and \ref{as:5.10} hold with $a=-1,4$, $b=4$ and $c=2$. Then, we have
\begin{align*}
  \frac{1}{\sqrt{n}}\sum_{t=1}^n    \begin{pmatrix}
       \hat{D}_t \big(\eta_t^{*2}-1\big)\\
   \big(\eta_t^* -\hat{\xi}_{n,\alpha}\big)\mathbbm{1}_{\{\eta_t^*< \hat{\xi}_{n,\alpha}\}}+\alpha_n \big(\hat{\xi}_{n,\alpha}+ \hat{\mu}_{n,\alpha}\big)
      \end{pmatrix}
\overset{d^*}{\to}N (0,\Psi_\alpha)
\end{align*}
almost surely with
\begin{align*}
    \Psi_\alpha = \begin{pmatrix}
      (\kappa-1)J& \alpha x_\alpha \Omega\\
    \alpha x_\alpha \Omega' & \alpha^2 \sigma_\alpha^2
      \end{pmatrix}
\end{align*}.
\end{lemma}
\begin{proof}
%
\citeauthor{beutner2018residual} (\citeyear{beutner2018residual}, proof of Lemma 7) shows that $\frac{1}{\sqrt{n}} \sum_{t=1}^n \hat{D}_t \big(\EE^*[\eta_t^{*2}]-1\big)=0$ for sufficiently large $n$ almost surely since $\hat{\theta}_n\overset{a.s.}{\to}\theta_0 \in \mathring{\Theta}$ and $\EE^*\big[\eta_t^{*2}\big]=1$ whenever $\hat{\theta}_n \in \mathring{\Theta}$ under Assumption \ref{as:5.10}. It remains to show that for each  $\lambda=(\lambda_1',\lambda_2)'  \in \R^{r+1}$ with $||\lambda||\neq 0$
%
\begin{align*}
 \sum_{t=1}^n \underbrace{\frac{1}{\sqrt{n}}\lambda'  \begin{pmatrix}
       \hat{D}_t \big(\eta_t^{*2}-\EE^*[\eta_t^{*2}]\big)\\
   \big(\eta_t^* -\hat{\xi}_{n,\alpha}\big)\mathbbm{1}_{\{\eta_t^*< \hat{\xi}_{n,\alpha}\}}+\alpha_n \big(\hat{\xi}_{n,\alpha}+ \hat{\mu}_{n,\alpha}\big)
      \end{pmatrix}}_{Z_{n,t}^*}
    \overset{d^*}{\to}N \left(0,\lambda'\Psi_\alpha \lambda\right)
\end{align*}
almost surely by the Cram\'er-Wold device. By construction, we have $\EE^*\big[Z_{n,t}^*\big]=0$. Further, we have that $s_n^2 =\sum_{t=1}^n\Var^*\big[Z_{n,t}^{*}\big]$ is equal to
\begin{align}
\lambda'\begin{pmatrix}
      \Var^*[\eta_t^{*2}]\hat{J}_n& \Cov^*[\eta_t^{*2},\big(\eta_t^* -\hat{\xi}_{n,\alpha}\big)\mathbbm{1}_{\{\eta_t^*< \hat{\xi}_{n,\alpha}\}}] \hat{\Omega}_n\\
    \Cov^*[\eta_t^{*2},\big(\eta_t^* -\hat{\xi}_{n,\alpha}\big)\mathbbm{1}_{\{\eta_t^*< \hat{\xi}_{n,\alpha}\}}] \hat{\Omega}_n' & \Var^*[\big(\eta_t^* -\hat{\xi}_{n,\alpha}\big)\mathbbm{1}_{\{\eta_t^*< \hat{\xi}_{n,\alpha}\}}] 
      \end{pmatrix}\lambda.
\end{align}
\citeauthor{beutner2018residual} (\citeyear{beutner2018residual}, Lemma 2) gives $\hat{J}_n \overset{a.s.}{\to}J$ and $\hat{\Omega}_n\overset{a.s.}{\to}\Omega$. Further, \citeauthor{beutner2018residual} (\citeyear{beutner2018residual}, Lemma 5) implies
\begin{align*}
&\Var^*[\big(\eta_t^* -\hat{\xi}_{n,\alpha}\big)\mathbbm{1}_{\{\eta_t^*< \hat{\xi}_{n,\alpha}\}}] \\
&=  \EE^*[\eta_t^{*2} \mathbbm{1}_{\{\eta_t^*< \hat{\xi}_{n,\alpha}\}}]- \Big(\EE^*[\eta_t^{*} \mathbbm{1}_{\{\eta_t^*< \hat{\xi}_{n,\alpha}\}}]\Big)^2 
+\hat{\xi}_{n,\alpha}^2 \EE^*[ \mathbbm{1}_{\{\eta_t^*< \hat{\xi}_{n,\alpha}\}}] \Big(1-\EE^*[ \mathbbm{1}_{\{\eta_t^*< \hat{\xi}_{n,\alpha}\}}]\Big)\\
& \ \ \  -2 \hat{\xi}_{n,\alpha} \bigg(\EE^*[\eta_t^* \mathbbm{1}_{\{\eta_t^*< \hat{\xi}_{n,\alpha}\}}]-\EE^*[\eta_t^* \mathbbm{1}_{\{\eta_t^*< \hat{\xi}_{n,\alpha}\}}]\: \EE^*[ \mathbbm{1}_{\{\eta_t^*< \hat{\xi}_{n,\alpha}\}}]\bigg)\\
&\overset{a.s.}{\to} \EE[\eta_t^2 \mathbbm{1}_{\{\eta_t< \xi_\alpha\}}]- \Big(\EE[\eta_t \mathbbm{1}_{\{\eta_t< \xi_\alpha\}}]\Big)^2+ \xi_\alpha^2\EE[ \mathbbm{1}_{\{\eta_t< \xi_\alpha\}}] \Big(1-\EE[ \mathbbm{1}_{\{\eta_t< \xi_\alpha\}}]\Big)\\
& \ \ \ \ -2 \xi_\alpha \bigg(\EE[\eta_t \mathbbm{1}_{\{\eta_t< \xi_\alpha\}}]-\EE[\eta_t \mathbbm{1}_{\{\eta_t< \xi_\alpha\}}]\: \EE[ \mathbbm{1}_{\{\eta_t< \xi_\alpha\}}]\bigg)\\
%
&= \Var[\big(\eta_t -\xi_\alpha\big)\mathbbm{1}_{\{\eta_t< \xi_\alpha\}}]
\end{align*}
and
\begin{align*}
&\Cov^*[\eta_t^{*2},\big(\eta_t^* -\hat{\xi}_{n,\alpha}\big)\mathbbm{1}_{\{\eta_t^*< \hat{\xi}_{n,\alpha}\}}]\\
&= \EE^*[\eta_t^{*3}\mathbbm{1}_{\{\eta_t^*< \hat{\xi}_{n,\alpha}\}}]- \EE^*[\eta_t^{*2}]\:\EE^*\big[\eta_t^* \mathbbm{1}_{\{\eta_t^*< \hat{\xi}_{n,\alpha}\}}\big]\\
& \ \ \ \ -\hat{\xi}_{n,\alpha} \bigg( \EE^*[\eta_t^{*2}\mathbbm{1}_{\{\eta_t^*< \hat{\xi}_{n,\alpha}\}}]-\EE^*\big[\eta_t^{*2}\big]\:\EE^*\big[\mathbbm{1}_{\{\eta_t^*< \hat{\xi}_{n,\alpha}\}}\big]\bigg)\\
&\overset{a.s.}{\to} \EE[\eta_t^{3}\mathbbm{1}_{\{\eta_t< \xi_\alpha\}}]- \EE[\eta_t^{2}]\:\EE\big[\eta_t \mathbbm{1}_{\{\eta_t< \xi_\alpha\}}\big] -\xi_\alpha \bigg(\EE[\eta_t^{2}\mathbbm{1}_{\{\eta_t< \xi_\alpha\}}]-\EE\big[\eta_t^{2}\big]\:\EE\big[\mathbbm{1}_{\{\eta_t< \xi_\alpha}\}\big] \bigg)\\
%
%
&= \Cov[\eta_t^{2},\big(\eta_t -\xi_\alpha\big)\mathbbm{1}_{\{\eta_t< \xi_\alpha\}}]
\end{align*}
as well as
\begin{align*}
    \Var^*\big[\eta_t^{*2}\big] &= \EE^*\big[\eta_t^{*4}\big]-\Big(\EE\big[\eta_t^{*2}\big]\Big)^2\overset{a.s.}{\to}\kappa-1.
\end{align*}
Thus, we get $s_n^2\overset{a.s.}{\to} \lambda'\Psi_\alpha \lambda$. 
Next, we verify Lindeberg condition. 
%
For any $\varepsilon>0$ 
\begin{align*}
&  \sum_{t=1}^n  \EE^*\big[Z_{n,t}^{*2}\mathbbm{1}_{\{|Z_{n,t}^{*}|\geq s_n \varepsilon \}}\big] \leq  \underbrace{\sum_{t=1}^n\EE^*\big[Z_{n,t}^{*2}\mathbbm{1}_{\{|\eta_t^{*}|> C \}}\big]}_{I}+ \underbrace{\sum_{t=1}^n  \EE^*\big[Z_{n,t}^{*2}\mathbbm{1}_{\{|Z_{n,t}^{*}|\geq s_n \varepsilon \}}\mathbbm{1}_{\{|\eta_t^{*}|\leq C \}}\big]}_{II}
\end{align*}
holds, where $C>0$. Employing the elementary inequalities $(x+y)^z\leq 2^z(x^z+y^z) $ and $|x-y|^z\leq x^z+y^z$ for all $x,y,z\geq 0$ we find that
\begin{align*}
Z_{n,t}^{*2} \leq & \frac{8}{n} \Big(\big(\lambda_1'\hat{D}_t\big)^2 \big(\eta_t^{*4}+\EE^*[\eta_t^{*2}]^2\big)+
    \lambda_2^2 \big(\eta_t^{*2}+\hat{\xi}_{n,\alpha}^2+\hat{\mu}_{n,\alpha}^2\big)\Big).
\end{align*}
Thus, we obtain
\begin{align*}
I \leq &  \frac{8}{n}\sum_{t=1}^n\EE^*\bigg[ \Big(\big(\lambda_1'\hat{D}_t\big)^2 \big(\eta_t^{*4}+\EE^*[\eta_t^{*2}]^2\big)+
    \lambda_2^2 \big(\eta_t^{*2}+\hat{\xi}_{n,\alpha}^2+\hat{\mu}_{n,\alpha}^2\big)\Big)\mathbbm{1}_{\{|\eta_t^{*}|> C \}}\bigg]\\
= & 8 \Big( \lambda_1' \hat{J}_n \lambda_1 \EE^*\big[\eta_t^{*4}\mathbbm{1}_{\{|\eta_t^{*}|> C \}}\big]+\lambda_2^2 E^*\big[\eta_t^{*2}\mathbbm{1}_{\{|\eta_t^{*}|> C \}}\big]\\
&\qquad +\big(\lambda_1' \hat{J}_n \lambda_1 \EE^*[\eta_t^{*2}]^2+
    \lambda_2^2(\hat{\xi}_{n,\alpha}^2+\hat{\mu}_{n,\alpha}^2)\big)\EE^*\big[\mathbbm{1}_{\{|\eta_t^{*}|> C \}}\big]\Big)\\
    \overset{a.s.}{\to}&8 \Big( \lambda_1' J \lambda_1 \EE\big[\eta_t^{4}\mathbbm{1}_{\{|\eta_t|> C \}}\big]+\lambda_2^2 E\big[\eta_t^{2}\mathbbm{1}_{\{|\eta_t|> C \}}\big]\\
&\qquad +\big(\lambda_1' J \lambda_1 \EE[\eta_t^{2}]^2+
    \lambda_2^2(\xi_{\alpha}^2+\mu_{\alpha}^2)\big)\EE\big[\mathbbm{1}_{\{|\eta_t|> C \}}\big]\Big)
\end{align*}
and choosing $C$ sufficiently large yields $I\overset{a.s.}{\to}0$. Given a value of $C$, we have
\begin{align*}
II\leq  &  \frac{8}{n}  \sum_{t=1}^n \EE^*\bigg[ \Big(\big(\lambda_1'\hat{D}_t\big)^2 \big(\eta_t^{*4}+\EE^*[\eta_t^{*2}]^2\big)+
    \lambda_2^2 \big(\eta_t^{*2}+\hat{\xi}_{n,\alpha}^2+\hat{\mu}_{n,\alpha}^2\big)\Big)\\
    & \qquad \qquad \quad \times \mathbbm{1}_{ \{ ||\lambda_1|| (\eta_t^{*2}+\EE^*[\eta_t^{*2}]) \max_{t} ||\hat{D}_t||+
    |\lambda_2|(|\eta_t^*|+|\hat{\xi}_{n,\alpha}|+|\hat{\mu}_{n,\alpha}|)\geq \sqrt{n} s_n \varepsilon \}}\mathbbm{1}_{\{|\eta_t^{*}|\leq C \}}\bigg]\\
\leq  &  \frac{8}{n}  \sum_{t=1}^n  \Big(\big(\lambda_1'\hat{D}_t\big)^2 \big(C^4+\EE^*[\eta_t^{*2}]^2\big)+
    \lambda_2^2 \big(C^2+\hat{\xi}_{n,\alpha}^2+\hat{\mu}_{n,\alpha}^2\big)\Big)\\
    & \qquad \qquad \quad \times \mathbbm{1}_{ \{ ||\lambda_1|| (C^2+\EE^*[\eta_t^{*2}]) \max_{t} ||\hat{D}_t||+
    |\lambda_2|(C+|\hat{\xi}_{n,\alpha}|+|\hat{\mu}_{n,\alpha}|)\geq \sqrt{n} s_n \varepsilon \}}\\
\leq  &  8  \Big(\lambda_1' \hat{J}_n \lambda_1 \big(C^4+\EE^*[\eta_t^{*2}]^2\big)+
    \lambda_2^2 \big(C^2+\hat{\xi}_{n,\alpha}^2+\hat{\mu}_{n,\alpha}^2\big)\Big)\\
    & \qquad \qquad \quad \times \mathbbm{1}_{ \{ ||\lambda_1|| (C^2+\EE^*[\eta_t^{*2}]) \max_{t} ||\hat{D}_t||+
    |\lambda_2|(C+|\hat{\xi}_{n,\alpha}|+|\hat{\mu}_{n,\alpha}|)\geq \sqrt{n} s_n \varepsilon \}}\\    
    \overset{a.s.}{\to}& 8  \Big(\lambda_1' J \lambda_1 \big(C^4+\EE[\eta_t^{2}]^2\big)+
    \lambda_2^2 \big(C^2+\xi_{\alpha}^2+\mu_{\alpha}^2\big)\Big)\times 0 = 0
\end{align*}
as $\max_t||\hat{D}_t||/\sqrt{n}\overset{a.s.}{\to}0$. Combining results, gives $\frac{1}{s_n^2}\sum_{t=1}^n  \EE^*\big[Z_{n,t}^{*2}\mathbbm{1}_{\{|Z_{n,t}^{*}|\geq s_n \epsilon \}}\big] \overset{a.s.}{\to}0$.
 The Central Limit Theorem for triangular arrays (c.f.\ \citeauthor{billingsley1986probability}, \citeyear{billingsley1986probability}, Theorem 27.3) implies that $\sum_{t=1}^n Z_{n,t}^*$ converges in conditional distribution to $N\big(0,\lambda'\Psi_\alpha\lambda \big)$ almost surely, which completes the proof. 
\end{proof}

\subsection{Derivation of Analytical Expressions}
\label{app:5.B}

Let $Y \sim t_{\nu}$ with cdf $F_{\nu}$ and pdf $f_{\nu}$, where $t_{\nu}$ denotes the Student-$t$ distribution with $\nu$ degrees of freedom. Define $\sigma^{2}_{\nu} = \frac{\nu-2}{\nu}$ such that $\eta = \sigma_{\nu}Y$ is now appropriately standardized such that $\EE[\eta^{2}] = 1$. Then we get
\begin{align*}
&F(x) = \PP(\eta \leq x) = \PP \left(Y \leq \sigma^{-1}_{\nu} x \right) = F_{\nu} \left( \sigma^{-1}_{\nu} x \right) \\    
&\xi_{\alpha} = F^{-1}(\alpha) = \sigma_{\nu} F^{-1}_{\nu} (\alpha).    
\end{align*}
The following relationship links the (conditional) moments of $\eta$ and $Y$: 
\begin{equation}
\label{eq:5.B.1}
\EE [\eta^{m} | \eta < \xi_{\alpha}] = \sigma_{\nu}^{m} \EE [Y^{m} | Y < F^{-1}_{\nu} (\alpha)]
\end{equation}
with $m \in \N$. Using moments of the truncated Student-$t$ distribution derived in \citeauthor{kim2008} (\citeyear{kim2008}, p. 84) we can find closed form expressions for the conditional expectations of $Y^m$. For $m=1$ and any $b \in \R$ we have\footnote{We only truncate from above, hence the lower truncation bound of \cite{kim2008} is $a=-\infty$.}
\begin{align*}
\EE \left[ Y | Y < b \right] 
&= -\frac{\Gamma(\frac{\nu-1}{2}) \nu^{\nu/2}}{2 F_{\nu}(b) \Gamma(\frac{\nu}{2}) \sqrt{\pi}} (\nu + b^2)^{-\frac{\nu-1}{2}}
= -\frac{\Gamma(\frac{\nu-1}{2}) \sqrt{\nu}}{2 F_{\nu}(b) \frac{\nu-2}{2}\Gamma(\frac{\nu-2}{2}) \sqrt{\pi}} \left( 1 + \frac{b^2}{\nu} \right)^{-\frac{\nu-1}{2}} \\
&= -\frac{\Gamma(\frac{\nu-1}{2}) }{ F_{\nu}(b) \sigma_\nu \sqrt{\nu-2}\Gamma(\frac{\nu-2}{2}) \sqrt{\pi}} \left( 1 + \frac{(\sigma_\nu b)^2}{\nu-2} \right)^{-\frac{\nu-1}{2}} = - \frac{f_{\nu-2} \left( \sigma_{\nu} b \right)}{ F_{\nu} (b)\sigma_{\nu} },
\end{align*}
where we recognize that $\Gamma\big(\frac{\nu}{2}\big) = \frac{\nu-2}{2}\Gamma\big(\frac{\nu-2}{2}\big)$. Together with \eqref{eq:5.B.1} we have
\begin{align}
\label{eq:5.B.2}
\EE [\eta | \eta < \xi_{\alpha}] = \sigma_{\nu} \EE \left[ Y | Y < F_{\nu}^{-1} (\alpha) \right] = -\frac{f_{\nu-2} \left( \xi_{\alpha}\right)}{\alpha}.
\end{align}
Similarly, we can derive for $m=2$ and any $b \in \R$
\begin{align*}
\EE [Y^{2} | Y < b] 
&= \frac{1}{\sigma^{2}_{\nu}} - b \frac{\Gamma(\frac{\nu-1}{2})\nu^{\nu/2}}{2F_{\nu}(b)\Gamma\left(\frac{\nu}{2}\right)\sqrt{\pi}} \left(\nu + b^2 \right)^{-\frac{\nu-1}{2}} 
= \frac{1}{\sigma^{2}_{\nu}} + b \EE[Y | Y < b].     
\end{align*}
Combined with \eqref{eq:5.B.1} we arrive at 
\begin{align*}
\EE [\eta^2 | \eta < \xi_{\alpha}] 
&= \sigma^{2}_{\nu} \EE [Y^{2} | Y < F_{\nu}^{-1} (\alpha)] 
= \sigma^{2}_{\nu} \left( \frac{1}{\sigma^{2}_{\nu}} + F^{-1}_{\nu}(\alpha) \EE[Y | Y < F^{-1}_{\nu} (\alpha)] \right) \\
&= 1 + \sigma^{2}_{\nu} F^{-1}_{\nu} (\alpha) \EE[Y | Y < F^{-1}_{\nu} (\alpha)]  
= 1 -  \xi_{\alpha} \frac{f_{\nu-2} (\xi_{\alpha})}{\alpha}.
\end{align*}
Now, as $\PP \left[ \eta < \xi_{\alpha} \right] = \alpha$, we obtain
\begin{align}
\label{eq:5.B.3}
   \EE \left[ \eta^{2} \mathbbm{1}_{\{\eta < \xi_{\alpha}\}} \right] - \alpha = \EE \left[ \eta^{2} |\eta < \xi_{\alpha} \right] \PP \left[ \eta < \xi_{\alpha} \right] -\alpha = - \xi_{\alpha} f_{\nu-2} (\xi_{\alpha}).
\end{align}
Finally, we consider $m=3$. Using $\Gamma(\frac{\nu-1}{2}) = \frac{\nu-3}{2}\Gamma(\frac{\nu-3}{2})$ or equivalently $\Gamma(\frac{\nu-3}{2}) = \frac{2}{\nu-3}\Gamma(\frac{\nu-1}{2})$, it follows that 
\begin{align*}
\EE [Y^3 | Y < b] 
&= - \frac{\Gamma(\frac{\nu-3}{2})\nu^{\nu/2}}{2F_{\nu}(b) \Gamma\left(\frac{\nu}{2}\right)\sqrt{\pi}} (\nu + b^2)^{-\frac{\nu-3}{2}} - b^{2} \frac{\Gamma(\frac{\nu-1}{2})\nu^{\nu/2}}{2F_{\nu}(b)\Gamma\left(\frac{\nu}{2}\right)\sqrt{\pi}} (\nu + b^2)^{-\frac{\nu-1}{2}} \\
&= - \left( \frac{2(\nu + b^2)}{\nu-3}+b^2 \right) \frac{\Gamma(\frac{\nu-1}{2})\nu^{\nu/2}}{2F_{\nu}(b)\Gamma\left(\frac{\nu}{2}\right)\sqrt{\pi}} (\nu+b^2)^{-\frac{\nu-1}{2}}\\
&= \left(\frac{2(\nu+b^2)}{\nu-3} + b^{2} \right) \EE[Y | Y < b],
\end{align*}
which leads to
\begin{align*}
\EE [\eta^{3} | \eta < \xi_{\alpha}]
&= \sigma^{3}_{\nu} \EE[Y^{3} | Y < F_{\nu}^{-1} (\alpha)] 
= \sigma_{\nu}^{3} \left(\frac{2(\nu + \sigma_{\nu}^{-2}\xi_{\alpha}^{2})}{\nu-3} + \left( \frac{\xi_{\alpha}}{\sigma_{\nu}} \right)^{2} \right) \EE[Y | Y < F^{-1}_{\nu} (\alpha)] \\
&= \left( \frac{2(\nu\sigma_{\nu}^{2} + \xi_{\alpha}^{2})}{\nu-3} + \xi_{\alpha}^{2} \right)
\sigma_{\nu}\EE[Y | Y < F^{-1}_{\nu} (\alpha)] \\
&= -\left( \frac{2(\nu\sigma_{\nu}^{2} + \xi_{\alpha}^{2})}{\nu-3} + \xi_{\alpha}^{2} \right)
\frac{f_{\nu-2}(\xi_{\alpha})}{\alpha}.
\end{align*}
Thus, we obtain
\begin{align}
\label{eq:8732534}
\EE \left[ \eta^{3} \mathbbm{1}_{\{\eta < \xi_{\alpha}\}} \right] =
\EE \left[ \eta^{3} |\eta < \xi_{\alpha} \right] \PP \left[ \eta < \xi_{\alpha} \right]
= -\left( \frac{2(\nu\sigma_{\nu}^{2} + \xi_{\alpha}^{2})}{\nu-3} + \xi_{\alpha}^{2} \right)
f_{\nu-2}(\xi_{\alpha}).
\end{align}
From \eqref{eq:5.B.2}--\eqref{eq:8732534}, we get the following expressions for the quantities $\mu_{\alpha}$, $p_{\alpha}$ and $q_{\alpha}$:
\begin{align*}
\mu_{\alpha} &= -\EE [\eta | \eta < \xi_{\alpha}] = \frac{f_{\nu-2}(\xi_{\alpha})}{\alpha} \\
p_{\alpha} &= \EE \left[ \eta^{2} \mathbbm{1}_{\{\eta < \xi_{\alpha}\}} \right] - \alpha = - \xi_{\alpha} f_{\nu-2} (\xi_{\alpha}) \\
q_{\alpha} &= \EE \left[ \eta^{3} \mathbbm{1}_{\{\eta < \xi_{\alpha}\}} \right] =
-\left( \frac{2(\nu\sigma_{\nu}^{2} + \xi_{\alpha}^{2})}{\nu-3} + \xi_{\alpha}^{2} \right)f_{\nu-2}(\xi_{\alpha}).
\end{align*}
Note that the Student-$t$ distribution approaches the standard normal distribution as $\nu \rightarrow \infty$. In that case, $\sigma_{\nu} \rightarrow 1$ and also $f_{\nu}(\cdot) \rightarrow \phi (\cdot)$ and $F_{\nu}(\cdot) \rightarrow \Phi(\cdot)$, i.e.\ the standard normal pdf and cdf, respectively. Hence, when $\eta$ is standard normally distributed, we have
$\xi_{\alpha} = \Phi^{-1} (\alpha)$ as well as $\mu_{\alpha} =  \frac{\phi(\xi_{\alpha})}{\alpha}$, $p_{\alpha} = -\xi_{\alpha}\phi(\xi_{\alpha})$ and $q_{\alpha} = -(2+\xi_{\alpha}^{2})\phi(\xi_{\alpha})$.

\newpage

\singlespacing
\bibliographystyle{Chicago-modified}

\end{document}